\documentclass[12pt]{article}

\usepackage[left=2cm,top=2.5cm,right=2cm,bottom=3cm,nohead]{geometry}
\usepackage{amsmath,amsfonts,amsthm,amssymb, bm, bbm, color}
\usepackage{epsfig}
\usepackage{verbatim,epic, eepic}
\input amssym.def
\input amssym.tex
\usepackage{yfonts}[1998/10/03]

\newfont{\go}{ygoth.tfm scaled 1200}  

\newtheorem{theorem}{Theorem}[section]
\newtheorem{lemma}[theorem]{Lemma}

\newcommand{\abs}[1]{\left| #1 \right|}

\newcommand{\vect}[1]{\bm{#1}}
\newcommand{\R}{\mathbb{R}}

\newcommand{\dir}{\displaystyle{\not} D}

\newcommand{\Ht}{{\mathbb{H}^3}}

\newcommand{\ket}[1]{\left | #1 \right \rangle}

\newcommand{\dif}[2]{\frac{d #1}{d #2}}
\newcommand{\pd}[2]{\frac{\partial#1}{\partial#2}}
\newcommand{\pdd}[2]{\frac{\partial^2 #1}{\partial #2 ^2}}

\newcommand{\db}[2]{ \left \{ {#1}, {#2} \right \} }

\newcommand{\hopt}{h_{\mathrm{opt.}}}

\newcommand{\gopt}{g_{\mathrm{opt.}}}

\numberwithin{equation}{section}
\def\ben{\begin{equation}}
\def\een{\end{equation}}
\def\bea{\begin{eqnarray}}
\def\eea{\end{eqnarray}}

\begin{document}

\hfuzz=100pt
\title{Universal properties of the near-horizon optical geometry}
\author{G W Gibbons\footnote{g.w.gibbons@damtp.cam.ac.uk}\ \ \& C M Warnick\footnote{c.m.warnick@damtp.cam.ac.uk}}
\date{}

\maketitle

\vspace{-.6cm}
{ \small
\centerline{DAMTP, University of Cambridge,}
\centerline{Wilberforce Road, Cambridge}
\centerline{CB3 0WA, UK}}

\vspace{-7cm}
\begin{flushright}
\small
DAMTP-2008-80
\end{flushright}
\vspace{6cm}

\begin{abstract}
We make use of the fact that the optical geometry near a static non-degenerate Killing horizon is asymptotically hyperbolic to investigate universal features of black hole physics. We show how the Gauss-Bonnet theorem allows certain lensing scenarios to be ruled in or out. We find rates for the loss of scalar, vector and fermionic `hair' as objects fall quasi-statically towards the horizon. In the process we find the Li\'enard-Wiechert potential for hyperbolic space and calculate the force between electrons mediated by neutrinos, extending the flat space result of Feinberg and Sucher. We use the enhanced conformal symmetry of the Schwarzschild and Reissner-Nordstr\"om backgrounds to re-derive the electrostatic field due to a point charge in a simple fashion.
\end{abstract}
\vspace{.3cm}

\section{Introduction}

There  has been over the past few years a very large amount of theoretical
work on black holes addressing  problems in quantum gravity, supergravity,
string theory and M-theory. Typically one seeks  
solutions of the supergravity equations  in four or higher
dimensions, and while many are broadly similar to to the well known
Kerr-Newman-de-Sitter  family  in four spacetime dimensions, there are
many differences of detail and in higher dimensions qualitatively 
different features can arise. It is desirable therefore to 
fix upon universal properties, true for a broad class of black holes.
For this reason the near horizon  geometry of  extreme black holes
has received a great deal  of attention, since it universally
behaves like $AdS_2 \times M_{n-2}$, where $M_{n-2}$ is typically an
$n-2$ dimensional  Einstein space and the symmetry is enhanced
from $\mathbb{R}$ to $SO(2,1)$.  By contrast the universal
near horizon optical geometry of non-extreme horizons 
with its enhanced conformal symmetry has 
largely been ignored (notable exceptions are \cite{Sachs:2001qb, Haba:2007mk}). 
In other words, little has been done
to exploit the fact that near a non-extreme horizon of a static black hole 
with metric
\ben
ds ^2 =- V^2dt ^2 + \gamma _{ij} dx ^i dx ^j\,,
\een
$i=1,2,\dots,n-1$     
the optical metric
\ben
a_{ij} dx ^i dx ^j = V^{-2} \gamma_{ij}dx^i dx ^j\,,
\een   
becomes asymptotically hyperbolic, with a conformal boundary
whose geometry is that of the event horizon.
In the spherically symmetric case, the limiting optical geometry is precisely
that of hyperbolic space $H^{n-1} = SO(n-1,1)/SO(n-1)$ 
with radius of curvature equal to $\kappa ^{-1}$, where $\kappa$ is
the surface gravity.

This is especially ironic because   asymptotically hyperbolic geometry
has been  studied for some time  because of  the light it throws on
the no-hair properties of 
asymptotically  de-Sitter metrics and the freezing of perturbations
which  have crossed the horizon of  an inflationary universe
(see  \cite{Gibbons:2007fd} for a recent discussion and references to earlier work).
A much better known case arises in the AdS/CFT correspondence
where the asymptotically hyperbolic  
geometry of $AdS_n$ or, in its ``Euclidean'' formulation, $H^n$ is  
of interest.   

The aim of the present paper is to fill this gap by
embarking on an exploration of what can be learned about
the universal qualitative properties of black holes 
from studying their  near horizon optical geometry
using the tools of hyperbolic geometry.
We shall principally  be concerned with the 
two topics

\begin{itemize}
\item A  qualitative study of   of null geodesics
near a static  horizon using the Gauss-Bonnet theorem, rather
in the style of \cite{Gibbons:1993cy} in the case of cosmic strings.

\item A study of the shedding of `hair' near static event horizons using propagators in hyperbolic space.

\end{itemize}

Of course, in the case of astrophysical black holes
the near horizon geometry has long been studied under the rubric
of the ``Membrane Paradigm'' \cite{Thorne:1986iy} and its Rindler like features
have been described. However this work, mainly concentrates on the planar
approximation to the horizon geometry   and does not make use of detailed
 concepts and ideas of hyperbolic geometry. Closer to what we are
interested in is the work of Haba \cite{Haba:2007mk} which considers scalar fields near a Killing horizon using an optical geometry approach and constructs approximate Green's functions in cases where the horizon is not necessarily spherical. This approach is more in tune with our philosophy of seeking universal properties. We will focus on spherical horizons and show that the enhanced symmetry present in this case make approximate propagators much simpler to construct. We remark that the universal nature of black hole absorption cross-section \cite{Das:1996we} has recently played an important r\^ole in the understanding of the ratio of shear viscosity to entropy density of conformal fluid in the AdS / CFT correspondence \cite{Policastro:2001yc}.

The paper will be organised as follows: we first define the optical
metric and explore some of its properties, including a study of light rays near an event horizon using the Gauss-Bonnet theorem. We will then present a
general argument based on the near horizon limit of the optical
geometry to estimate the rate of loss of `hair' as bodies fall towards
the black hole. Then we will show how the
optical metric allows one to find the fields due to static electric
and scalar charges in the Schwarzschild and Reissner-Nordstr\"om
backgrounds with very little calculation. This is a re-derivation of
results in the literature in a more coherent and direct way. We will
include equipotential plots for a charged particle approaching a black
hole, graphically demonstrating the no-hair result.

\section{Optical Metrics}

The optical metric may be thought of as the modern incarnation of an
idea dating back to Fermat in the 17th century. Fermat expressed the
laws governing reflection and refraction of light as what we would now
call an action principle. His `principle of least time' states that
the path taken by a ray of light is that which minimises the time
taken between the two points. This can be used to derive the more
familiar Snell's law and other optical laws.

In the case of light rays moving in a static background, with a given
choice of time coordinate $t$, we may take
this at face value and define the action for light rays in the metric
\begin{equation}
g = -a^2(\vect{x}) dt^2 +h_{ij}(\vect{x}) dx^i dx^j \label{static}
\end{equation}
to be
\begin{equation}
S = \int dt = \int \sqrt{a^{-2} h_{ij} \dif{x^i}{\lambda} \dif{x^j}{\lambda}}d\lambda,
\end{equation}
where we use the fact that null rays have $ds=0$. Extremizing this
action gives the unparameterised geodesics of the $3$ dimensional
Riemannian metric:
\begin{equation}
h_{\mathrm{opt.}} = a^{-2}(\vect{x}) h_{ij}(\vect{x}) dx^i dx^j.
\end{equation}
These unparameterised geodesics are the {\it light rays} and the
metric $\hopt$ is the {\it optical metric}. One may check that these
unparameterised geodesics indeed coincide with the projections of the
null geodesics of (\ref{static}) onto the spacelike surfaces
$t=\mathrm{const.}$ and so the light rays are the paths traced by
photons moving in this static space. The equivalence is clear by considering the
metric
\begin{equation}
g_{\mathrm{opt.}} = a^{-2} g = -dt^2 + \hopt;
\end{equation}
since the unparameterised null geodesics are conformally invariant
objects, the result follows. We will sometimes refer to the
ultra-static metric $g_{\mathrm{opt.}}$ as the optical metric also,
relying on context to distinguish it from $\hopt$. The optical metric
is not necessarily unique as it depends upon a choice of time
coordinate $t$. For metrics which admit more than one choice of $t$
there can be more than one optical metric. We shall see this in detail
in the case of anti-de Sitter space below.

It is not only statements about the null geodesics which are accessible via the optical metric. Many of the field equations of physics both classical and quantum behave well under conformal transformations and so we can make use of the universal nature of the near horizon optical geometry to study physics near a black hole (or cosmological) horizon.

\subsection{The Optical Metrics of de Sitter and anti-de Sitter}

\subsubsection{de Sitter}

We start with $3+1$ dimensional de Sitter as the timelike hyperboloid in
$\mathbb{E}^{4,1}$:
\begin{equation}
X^2 + Y^2 + Z^2 + W^2 - V^2 = 1, \qquad ds^2 = dX^2 + dY^2 + dZ^2 + dW^2 - dV^2.
\end{equation}
We could consider $n+1$ dimensions, but the generalisations are
straightforward. A choice of static time coordinate $t$ corresponds to a choice of
future-directed, timelike, hypersurface orthogonal Killing vector
$\pd{}{t}$. The Killing vectors of dS are those in $\mathbb{E}^{4,1}$
which generate rotations and boosts. A basis for the Killing vectors
is given by the hypersurface orthogonal vectors:
\begin{equation}
M_{\mu \nu} = X_{\mu}\pd{}{X^\nu} - X_{\nu}\pd{}{X^\mu}.
\end{equation}
Here $\mu, \nu$ are $\mathbb{E}^{4,1}$ indices. There is no Killing vector which is everywhere timelike, however the
Killing vector:
\begin{equation}
K = W\pd{}{V} + V\pd{}{W}
\end{equation}
is timelike and future directed in the region $\{W^2 -
V^2>0\}\cup\{W>0 \}$. Furthermore, any other choice of timelike
Killing vector is equivalent to $K$ under a Lorentz transformation. We can find a parameterisation of the
hyperboloid in this patch, such that $K = \pd{}{t}$ is a static
Killing vector as follows:
\begin{eqnarray}
X &=& r \sin \theta \sin \phi, \nonumber \\
Y &=& r \sin \theta \cos \phi, \nonumber \\
Z &=& r \cos \theta,   \nonumber\\
W &=& \sqrt{1-r^2} \cosh t, \nonumber\\
V &=& \sqrt{1-r^2} \sinh t.
\end{eqnarray}
On this patch, the metric takes the form
\begin{equation}
ds^2 = (1-r^2)\left (-dt^2 + \frac{dr^2}{(1-r^2)^2} + \frac{r^2}{1-r^2}(d\theta^2+\sin^2
\theta d\phi^2) \right),
\end{equation}
so that the optical metric may be seen to be the Beltrami metric on
Hyperbolic space. In fact these coordinates cover all of the Beltrami ball and so the
optical geometry of the static slicing of de Sitter is precisely
$\Ht$. The conformal infinity of the hyperbolic ball corresponds to
the Killing horizon on the hyperboloid at $W^2-V^2=0$ where $K$
becomes null. It is a general characteristic of Killing horizons that
the optical geometry approaches a constant negative curvature geometry
near the horizon.

\subsubsection{Anti-de Sitter}

The situation for AdS is somewhat more interesting than that for dS
because there exist three equivalence classes of timelike, future
directed, hypersurface orthogonal Killing vectors under the action of
$SO(3,2)$. To see this we take AdS to be a hyperboloid in
$\mathbb{E}^{3,2}$:
\begin{equation}
-W^2-V^2+X^2+Y^2+Z^2=-1, \qquad ds^2 = - dW^2 - dV^2+dX^2 + dY^2 + dZ^2.
\end{equation}
In a similar way to the case of dS, a basis for the Killing vectors is
given by:
\begin{equation}
M_{\mu \nu} = X_{\mu}\pd{}{X^\nu} - X_{\nu}\pd{}{X^\mu},
\end{equation}
where $\mu, \nu$ are $\mathbb{E}^{3,2}$ indices. Under a $SO(3,2)$
transformation, any Killing vector which is timelike somewhere on the hyperboloid
may be brought into one of three forms, listed below with the region
in which they are timelike:
\begin{equation}
\begin{array}{r c l l}
K_1 &=& V\pd{}{W}-W\pd{}{V}, \qquad &\mbox{all of AdS},  \\
K_2 &=& (Z+W)\pd{}{V} + V\left(\pd{}{Z}-\pd{}{W} \right), \qquad
&\{W+Z>0\}, \\
K_3 &=& Z\pd{}{V}+V\pd{}{Z}, \qquad &\{Z^2-V^2>0\} \cup \{Z>0\} .
\end{array}
\end{equation}
We now find the optical metric in each case:

{\bf Case 1} We pick the following parameterisation of the hyperboloid
\begin{eqnarray}
X &=& r \sin \theta \sin \phi, \nonumber \\
Y &=& r \sin \theta \cos \phi, \nonumber \\
Z &=& r \cos \theta,   \nonumber\\
W &=& \sqrt{1+r^2} \cos t, \nonumber\\
V &=& \sqrt{1+r^2} \sin t,
\end{eqnarray}
so that the metric is given by:
\begin{equation}
ds^2 = (1+r^2)\left (-dt^2 + \frac{dr^2}{(1+r^2)^2} + \frac{r^2}{1+r^2}(d\theta^2+\sin^2
\theta d\phi^2) \right), 
\end{equation}
and we recognise that the optical metric is the Beltrami metric for
$S^{3}$. This covers one half of the sphere, with the $2$-sphere at
$r=\infty$ corresponding to an equatorial $2$-sphere.

{\bf Case 2} We pick a different parameterisation for the hyperboloid:
\begin{eqnarray}
X &=& x/z, \nonumber\\
Y &=& y/z, \nonumber\\
Z &=& (1+t^2-x^2-y^2-z^2)/(2z), \nonumber \\
W &=& (1-t^2+x^2+y^2+z^2)/(2z), \nonumber \\
V &=& t/z,
\end{eqnarray}
so that the metric is the Poincar\'e upper half-space metric:
\begin{equation}
ds^2 = \frac{1}{z^2} \left(-dt^2+dx^2+dy^2+dz^2 \right), \qquad z>0
\end{equation}
and the optical geometry is the half-space $\{z>0\}$ in $\mathbb{E}^3$.

{\bf Case 3} Finally we consider the case where $\pd{}{t}=K_3$. A
suitable parameterisation of the static patch is provided by:
\begin{eqnarray}
X &=& \tan \theta \cos \phi, \nonumber\\
Y &=& \tan \theta \sin \phi, \nonumber\\
Z &=& \cosh t \hspace{.15cm}\mathrm{cosech} \chi \sec \theta, \nonumber\\
W &=& \mathrm{cotanh} \chi \sec \theta, \nonumber\\
V &=& \sinh t \hspace{.15cm} \mathrm{cosech} \chi \sec \theta.
\end{eqnarray}
The spatial coordinates have ranges $0\leq \phi < 2\pi$, $0\leq \theta <
\pi/2$, $0< \chi$. The metric is given by:
\begin{equation}
ds^2 = \mathrm{cosech}^2 \chi \sec^2 \theta \left( -dt^2 + d\chi^2 +
\sinh^2 \chi (d\theta^2+\sin^2 \theta d\phi^2)\right).
\end{equation}
We see that the optical metric is that of hyperbolic space in geodesic
polar coordinates. Since $\theta$ does not range over $[0, \pi)$ the
coordinates only cover half of $\Ht$ and there is a boundary which is
given by the plane $\theta = \pi/2$ in these coordinates.

We note that all three of the AdS optical metric have a finite
boundary. This boundary corresponds to the conformal infinity of the
AdS space and is a manifestation of the fact that AdS is not globally
hyperbolic. In the case of dS, the optical metric is complete and has
an asymptotically hyperbolic end which corresponds to the Killing
horizon of the static patch. As we will see below, this behaviour is
typical of a Killing horizon, such as the event horizons of
Schwarzschild and Reissner-Nordstr\"om.

\subsection{The Optical Metric of Schwarzschild and Reissner-Nordstr\"om}

Owing to the similarities between the cosmological horizon of de Sitter and the event horizons of black holes, we might expect the optical geometries to be similar near the horizon. We shall
see that this is indeed the case, and that near the horizon, the
geometry of a static black hole has an asymptotically hyperbolic optical
metric.

We start with the Schwarzschild metric
\begin{equation}
ds^2 = -\left(1-\frac{2M}{r}\right) dt^2 + \frac{dr^2}{1-\frac{2M}{r}} + r^2
(d\theta^2 + \sin^2 \theta d\phi^2) \label{schcoord}
\end{equation}
and make the coordinate transformation
\begin{equation}
r = M \left (\frac{1+\rho}{\rho}\right)
\end{equation}
this takes the asymptotically flat end to $\rho=0$ and the horizon to
$\rho = 1$ and puts the metric into the form:
\begin{equation}
ds^2 = \left(\frac{1-\rho}{1+\rho} \right) \left [ -dt^2 + 16 M^2
  \left(\frac{1+\rho}{2 \rho} \right)^4 \left\{
  \frac{d\rho^2}{(1-\rho^2)^2} + \frac{\rho^2}{1-\rho^2} \left(
  d\theta^2 + \sin^2\theta d \phi^2 \right) \right \} \right ]
\end{equation}
The term inside the braces may be seen to be the metric on $\Ht$ in
Beltrami coordinates. In the limit $\rho \to 1$, we thus see that the
optical metric tends to a metric of constant negative curvature as we
approach the horizon. 

The case of Reissner-Nordstr\"om is rather similar, although the
resulting metric is not so elegant. In the familiar coordinates, the
Reissner-Nordstr\"om metric is given by
\begin{equation}
ds^2 = -\left(1-\frac{2M}{r}+ \frac{Q^2}{r^2}\right) dt^2 + \frac{dr^2}{1-\frac{2M}{r}+ \frac{Q^2}{r^2}} + r^2
(d\theta^2 + \sin^2 \theta d\phi^2)
\end{equation}
In these coordinates, the horizon is at $r=M+\sqrt{M^2-Q^2} = M+\mu$,
where we define a new parameter $\mu$ which we will assume to be
strictly positive. The case $\mu=0$ corresponds to an extremal black hole which we
will not consider here. The coordinate transformation
\begin{equation}
r = M + \frac{\mu}{\rho}
\end{equation}
puts the metric into the form
\begin{equation}
ds^2 = \frac{\mu^2 (1-\rho^2)}{(\mu + m \rho)^2} \left [-dt^2 +
  \frac{(\mu+m)^4}{\mu^2} \left(\frac{\mu + m \rho}{(\mu+m)\rho}\right)^4\left\{
  \frac{d\rho^2}{(1-\rho^2)^2} + \frac{\rho^2}{1-\rho^2} \left(
  d\theta^2 + \sin^2\theta d \phi^2 \right) \right \} \right]
\end{equation}
We see once again that the optical metric approaches the Beltrami
metric on hyperbolic space as we get close to the horizon. In both
cases, the radius of the hyperbolic space is $\beta_H$, the inverse Hawking
temperature of the black hole. In fact, this is a general property of
a static metric with a non-degenerate Killing horizon as shown in
\cite{Sachs:2001qb}. In that paper Sachs and Solodukhin show that near a non-extreme horizon of a static black hole 
with metric
\ben
ds ^2 =- V^2dt ^2 + \gamma _{ij} dx ^i dx ^j\,,
\een
$i=1,2,\dots,n-1$     
the optical metric
\ben
a_{ij} dx ^i dx ^j = V^{-2} \gamma_{ij}dx^i dx ^j\,,
\een   
becomes asymptotically hyperbolic, with a conformal boundary
whose geometry is that of the event horizon. Essentially this is due to the fact that at a non degenerate Killing horizon $V^2$ must have a simple zero. We will consider the spherically symmetric case from here on, however we will try and identify results which we expect to remain the same in the case of more interesting (compact) horizon topology.

It will prove crucial in our exact calculations
later that the optical metrics of both Schwarzschild and Reissner-Nordstr\"om take the form:
\begin{equation}
g_{\mathrm{opt.}} = -dt^2 + H^4 h,
\end{equation}
where $h$ is the metric on the unit pseudo-sphere, $\Ht$, and $H$ is a
{\it harmonic} function on $\Ht$ which approaches $1$ near the
conformal boundary of $\Ht$. This observation is responsible for the fact that the fields due to
static electric and scalar charges in these backgrounds may be found
explicitly \cite{Copson:1928, Linet:1976sq, Leaute:1976sn}. We show below how these fields may be constructed. This special form of the metric occurs only in the $4$-dimensional space-times, so does not, unfortunately, lead to a generalisation of these results in an obvious way to higher dimensions.

For another viewpoint on the optical geometry of Schwarzschild see \cite{Abramowicz:2002qf} where the geometry is constructed as an embedding in a higher dimensional hyperbolic space.

\section{Lensing and The Gauss-Bonnet theorem}

In order to discuss null geodesics we could follow the well trodden path of solving the differential equations. Instead we will follow the approach of \cite{Gibbons:1993cy, Gibbons:2008rj} and extract information about geodesics using the Gauss-Bonnet theorem which directly involves the negative curvature of the optical metric. Although here we consider only the Schwarzschild black hole, it is clear that many qualitative properties may be deduced using only the assumption of negative curvature near the horizon of the optical metric, provided a totally geodesic $2$-surface exists.

Let us now consider geodesics lying  in an oriented 
two-surface $\Sigma$. We may apply the Gauss-Bonnet
theorem to obtain useful information \cite{Gibbons:1993cy}, in particular
about angle sums of geodesic triangles.  Let $D \subset \Sigma$ be
domain with Euler number $\chi(D) $ 
and a not necessarily connected 
boundary $\partial D$, possibly with corners at which the tangent
vector
of the boundary is discontinuous. If $K$ is the Gauss
curvature
of $D$, such that $R_{ijkl}=K(f_{ik}f_{jl}-f_{il}f_{jk} )
 $ and $k$ the curvature of $\partial D$, $\theta _i$ the angle
through which the tangent turns inwards at the i'th  corner  then
\ben 
\int_D  KdA + \oint _{\partial D} k dl + \sum _i \theta _i  = 2 \pi \chi(D)\,.
\een  

In the case of the Schwarzschild metric, if one 
considers geodesics in an equatorial plane
the optical metric is
\ben
ds ^2 = \frac{dr ^2 }{ (1-\frac{2M }{ r})^2}  + \frac{r^2 }{(1- \frac{2 M }{
    r} ) } d \phi ^2\,.  
\een
We return here to the standard Schwarzschild coordinates of (\ref{schcoord}). Note that the radial optical distance is 
\ben
\frac{ dr }{ (1-\frac{2M }{ r})}= dr^\star, 
\een
where $r^\star = r -2M+  2M \ln ( \frac{r  }{ 2M}-1)$ 
is the Regge-Wheeler tortoise coordinate. 

There is a circular geodesic at $r=3M$ and the horizon
$r=2M$ is at an infinite optical distance inside this
at $r^\star =-\infty$. The Gauss curvature 
\ben
K = -\frac{2M}{r^3} ( 1 -\frac{3 M }{ 2r})\,
\een
 is {\sl everywhere negative}. It falls to zero  like $-\frac{2 M }{ r^3}$
at infinity but near the horizon the Gauss-curvature approaches the
negative constant $-\frac{1 }{( 4M ) ^2}$. This is precisely as we
expect to find given the results of the previous section.

The fact that the Gauss curvature is {\sl negative}
looks on the face of it rather paradoxical, since one
usually thinks of gravitational fields as {\it focussing}
a bundle of light rays. However, as Lodge perhaps dimly realised \cite{Lodge}
a spherical vacuum gravitational field does not quite act in that way.
The equation of geodesic deviation governing the separation
 $\eta$ of two neighbouring light rays in the equatorial plane
is 
\ben
\frac{ d^2 \eta }{ dt ^2 } + K \eta =0\,.
\een        
Thus neighbouring light rays actually {\sl diverge}.
The focussing effect of a gravitational lens is not, as we shall see shortly,
a local but rather a global, indeed even topological, effect.

One might wonder whether the full 3-dimensional curvature 
$^{3}R_{ijkl} $ of the optical metric  has all of its sectional
curvatures negative, but this cannot be. The sectional curvature
of a surface is related to the full curvature tensor by
\ben
^{3}R_{ijkl}= K(f_{ik}f_{jl} -f_{1l}f_{jk} ) -K_{ik}K_{jl}+K_{il}g_{jk}\,,  
\een   
where $K_{ij}$ is the second fundamental form or extrinsic
curvature of the surface. For a totally geodesic surface
$K_{ij}=0$, and the two sectional curvatures agree.
One such totally geodesic surface is the equatorial
plane for which, as we have seen, $K$ is negative.
Another totally geodesic submanifold is the sphere at $r=3M$
for which $K $ is obviously positive.  

The negativity of the Gauss curvature of the optical metric
in the equatorial plane is a fairly {\it universal} property of black hole
metrics. To see this we note that if
\ben
ds ^2 = d \rho ^2 + l^2(\rho) d \phi ^2\,,
\een
then
\ben
K= - \frac{ 1 }{ l} \frac{ d^2 l }{ d \rho ^2 }\,.
\een
Any metric with the same qualitative features as the
Schwarzschild metric, as long as it has a {\sl positive} mass,  will
have $K$ negative. Indeed this fact might  be made the basis
of excluding negative mass objects observationally.    

A simple calculation shows the integral over the region
outside the circular geodesic at $r=3M$ is  
\ben
\int _{r\ge 3M} K dA = -2 \pi\,. \label{total}
\een

\begin{figure}
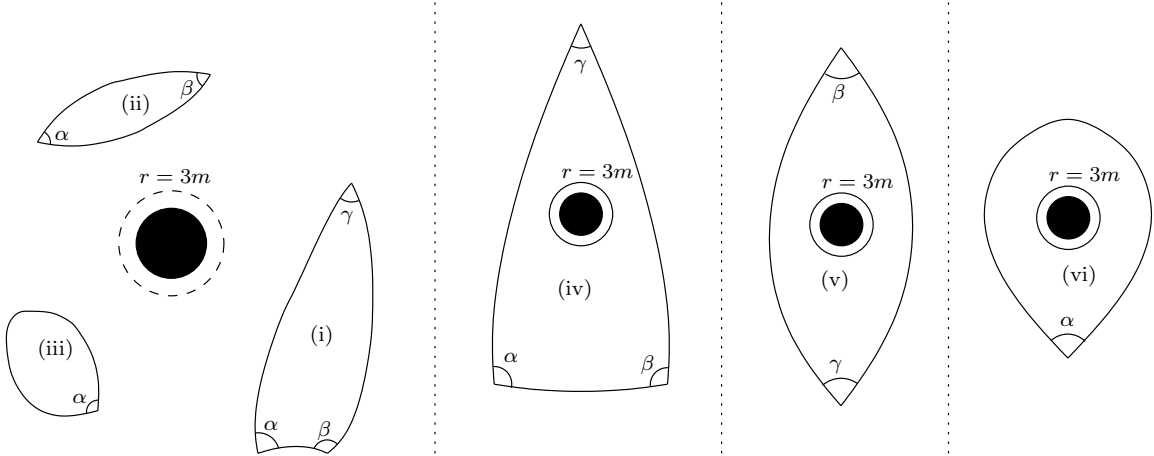

\centering {
\input{gb1.tex} \hspace{.5cm}
\input{gb2.tex}}
\caption{The geodesic polygons described in (i)-(vi)}
\end{figure}

Let us now apply the Gauss-Bonnet theorem to various cases.

\medskip \noindent (i)  Geodesic triangle $\Delta$ not containing the 
the region inside $r=3M$. In this case $\chi(\Delta)=1$.
If $\alpha, \beta, \gamma$ are the necessarily positive
internal angles, we find
that the angle sum is less that $\pi$,\ben
\alpha + \beta + \gamma = \pi+  \int _\Delta  KdA\, <  \pi \,. 
\een 

 \medskip \noindent (ii)  Geodesic di-gon  $ S$ not containing the 
the region inside $r=3M$. In this case $\chi(S)=1$.
 If $\alpha$ and $ \beta$ are the internal angles, 
\ben
\alpha + \beta  =   \int _S  KdA\, <  0 \,.
\een 
In other words two such geodesics cannot intersect  twice if the hole
is not inside the di-gon. Neither, in these circumstances,
 can a geodesic intersect itself because

\medskip \noindent (iii)  Geodesic loop  $T$ not containing the 
the region inside $r=3M$. In this case $\chi(T)=1$ and one finds
that if if $\alpha$ is the internal angles,the \ben
\alpha  = -\pi+  \int _T KdA\, < -\pi \,, 
\een 
which is plainly impossible.
 
This might seem counter-intuitive in the light of 
one's usual intuition about light bending,
but this feeling is dispelled
by considering cases in which the domain $D$ has two boundary
components,
the second, inner, one  being the circular geodesic at $r=3M$.
The domain with the circle removed has the topology
of an annulus and thus its Euler number vanishes.

\medskip \noindent (iv)  Geodesic triangle with hole  $\Delta _o$ 
enclosing the geodesic circle at $r=3M$ and 
with the region the region inside $r=3M$ removed.

If $\alpha, \beta, \gamma$ are the internal angles, we find
that the angle sum  is greater than $\pi$,\ben
\alpha + \beta + \gamma = 3\pi+  \int _{\Delta_0 }KdA  \ge \pi \,. 
\een 

Similarly

 \medskip \noindent (v)  Geodesic di-gon  $S_0$ with the 
the region inside $r=3M$ removed. In this case $\chi(S_0)=0$ and one finds
that if If $\alpha$ and $ \beta$ are the internal angles, 
then\ben
\alpha + \beta  =   2\pi + \int _{S_0}KdA\, >  0 \,.
\een 
In other words two such geodesics may  intersect  twice if the hole
is  inside the di-gon. Moreover, in these circumstances,
 a geodesic can intersect itself because:

\medskip \noindent (vi)  Geodesic loop  $T_0$  containing the 
the region inside $r=3M$. In this case $\chi(T _0)=1$ and 
if If $\alpha$ is the internal angle, we find
that \ben
\alpha  = \pi+  \int _{T _0}  K dA    \,, 
\een 
which is plainly possible.

Similar results may be obtained by considering geodesics inside
$r=3M$, but now domain must not contain the horizon,
otherwise $\int _D K dA$ will diverge.
Near the horizon the geometry is that of Lobachevsky space
with constant  curvature $- \frac{ 1 }{ 4M}$.
 
\medskip \noindent (vii) Deflection. We consider a geodesic line
with no self-intersection  which at
large distances, is radial. The angle between the asymptotes is $\delta$,
with the convention that it is positive if the light ray is bent towards
the hole. The geodesic decomposes the region inside two circles,
one  of very large radius and the other at $r=3M$ into
two domains $D_\pm$ whose common boundary component  consists of
the geodesic, which intersects the circle at infinity at right angles.
 We chose $D_+$ to enclose  the hole  so it has an inner boundary
component  at $r=3M$ and a portion of the circle at infinity through
which the angle $\phi$ has range $\pi-\delta $. 
Clearly $D_+$ is topologically an annulus and  so it has 
vanishing Euler number, $\chi(D_+)=0$.    
The other domain has Euler number $\chi(D_-)=1$,
 and $\phi$ ranges through
$\pi +\delta$. The Gauss-Bonnet formula applied to $D_\pm$ 
acquires a contribution
from  the two corners  and the circle at infinity.
The result is
\ben
\delta = -\int _{D_-} K dA >0.
\een
\begin{figure}
\centering {
\input{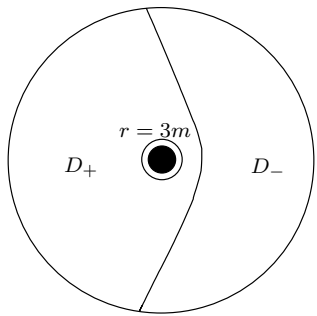} }
\caption{Light bending by a Schwarzschild black-hole}
\end{figure}

For a geodesic whose distance of closest approach is very large,
we may estimate this integral by approximating the geodesic as the 
 straight line $r=\frac{ b }{ \sin \phi} $. The impact parameter 
to this lowest non-trivial order coincides with the distance of 
 nearest approach and equals $b$. To the necessary accuracy
\ben
K dA \approx -\frac{ 2 M }{ r^3} rdr d \phi \,.
\een  

The domain of integration  $D_-$  is, with sufficient accuracy
 over $r\ge \frac{b }{ \sin \phi}$,
$0\le \phi \le 2\pi$. A simple calculation gives
the classic result
\ben
\delta = \frac{ 4 M }{ b}  \,.
\een

Note the same method works for any static metric,
not just Schwarzschild (for an application to gravitational lensing see \cite{Gibbons:2008rj}) and shows that the Gauss-Bonnet method
does not just give qualitative results, but it  can
 be made into a quantitative tool.

\section{No-hair properties from the optical metric}

We now shift our attention to a different aspect of black hole physics, 
the so called `no-hair' property. A stationary
black hole has only three measurable quantities associated with it:
mass, angular momentum and electric charge. Thus no matter how
complicated a system we start with, once it has undergone
gravitational collapse to form a black hole, we are left with only
these three pieces of information. This presents only a minor
philosophical problem if we are prepared to accept that the information
contained in the initial system is somehow trapped irretrievably behind
the horizon. Once one includes Hawking's observation that black holes
may radiate and indeed evaporate over time, the question of where the
information goes becomes more vexed, giving rise to the so called
`information loss paradox'.

We shall be interested with discovering how, at the classical level,
information is lost as a body falls into a black hole. We will work
with an approximation where the in-falling body is supposed to have a
negligible effect on the background and so we may consider physics in
a fixed black-hole geometry. This amounts to a linearisation of the problem, but allows
analytic progress to be made. 

At the linearised level, the no-hair property may be translated
mathematically into the notion that the black hole exterior cannot
support any external fields which are regular both at the horizon and
spacelike infinity. An interesting question is what happens to the
fields around some compact body as it falls into the black hole. This
corresponds to asking what happens to a propagator as its pole
approaches the horizon. This gives information both about the
classical scenario, but also about the outcome of scattering
experiments performed as the body falls towards the hole 
\cite{Teitelboim:1972qx}.

We are interested in finding a fairly general approach to study how
`hair' is lost as bodies carrying charges fall into a black hole. We will argue that this is a
property of the geometry close to the black hole horizon, and so we
can consider the problem in this region where the geometry
simplifies. We translate the question of finding propagators to a
problem in the optical metric and then show how we can estimate the
rate of information loss in this geometry. As we noted above, near the horizon this geometry approaches that of hyperbolic space irrespective of the details of the black hole under consideration.

\subsection{Physics in $\mathbb{R}_t\times \Ht$ \label{phys}}

There have been investigations of physics in spaces of constant negative curvature for some time and with varying motivations. Callan and Wilczek initiated a study of quantum mechanics on $\mathbb{H}^4$ in \cite{Callan} in order to geometrically regulate the infra-red divergences of Euclidean field theory. In \cite{Atiyah} Atiyah and Sutcliffe considered Skyrmions in $\Ht$ as a means of finding approximate Skyrmions in $\mathbb{E}^3$ for the case where the pion mass is non-zero. Field theories on $\Ht$ are also thermodynamically interesting as one might expect, anticipating the Hawking radiation of horizons. A study of some thermodynamic properties, especially Bose-Einstein condensation is shown in \cite{Cognola}. There have also been studies of electrostatics and magnetostatics in hyperbolic space, with particular reference to the Gauss linking formula \cite{deTurck}.

Although the references above provide a reasonably comprehensive discussion of physics in  $\mathbb{R}_t\times \Ht$, in the interests of a self-contained exposition we will discuss some aspects here. Using the near horizon limit of the optical metric, this corresponds after a conformal transformation to physics in the neighbourhood of a non-extremal black hole horizon. The fact that analytic progress is possible may be traced to the fact that the metric is conformally
equivalent to the inside of the future light-cone of the origin in Minkowski space
as follows.

Throughout this section we represent a point in $\Ht$ as a point on the unit pseudosphere
$\Ht=\{X\cdot X=-1, X_0>0\}$ in $\mathbb{E}^{3,1}$. This makes the
equations manifestly $SO(3,1)$ invariant and easy to translate between
different coordinate systems. The point $(t,X)$ in  $\mathbb{R}_t\times \Ht$ is mapped to a point on the interior of the future light
cone of the origin in $\mathbb{E}^{3,1}$ according to:
\begin{eqnarray}
\phi : \mathbb{R}_t\times \Ht &\to& \{x\in \mathbb{E}^{3,1}, x\cdot
x=-1, x_0>0\} \nonumber\\
(t,X) &\mapsto& x = X e^{t}.
\end{eqnarray}
If $g$ is the metric on  $\mathbb{R}_t\times \Ht$and $\eta$ is the standard metric on $\mathbb{E}^{3,1}$, then one finds
that:
\begin{equation}
\phi_* \eta = e^{2 t} g,
\end{equation}
so we have exhibited the conformal equivalence of these two
spaces. This means that given any conformally invariant equation whose
propagator may be found in Minkowski space, one may find the
propagator for $\mathbb{R}_t\times \Ht$.

\subsubsection{Massless Wave Equation}

Although the massless wave equation is conformally invariant, and so
the propagator may be constructed from
the known flat space propagator, it is more convenient to directly
solve in this case. We seek to solve the equation:
\begin{equation}
(\Box_g-\frac{1}{6}R_{g}) G(t,X;\tau, Y) = \left(- \partial_t^2 + \Delta_h+1\right)G= \delta^{(4)}_g((t,X), (\tau,Y)).
\end{equation}
One might think that the appearance of the curvature term above gives rise to an effective mass, however it is important to include this term in the massless wave equation to ensure, for example, that disturbances propagate along the light-cone as one would expect. Following standard treatments, one Fourier transforms in time and
takes $\tau=0$ without loss of generality. We then need to solve
the Helmholtz equation
\begin{equation}
\left( \Delta_h + 1+k^2\right)\tilde{G} = \delta_h (X, Y).
\end{equation}
This has the general solution, found by using geodesic polar coordinates on $\Ht$:
\begin{equation}
\tilde{G}(X,Y) = \frac{A e^{i k \chi}+ Be^{-i k \chi}}{4 \pi \sinh \chi}
\end{equation}
where $\chi = D(X,Y)$ and  $A+B=1$. The fact that this (and other Green's functions on $\Ht$) depends only on $D(X,Y)$ is due to the $2$-point homogeneity of the space. Undoing the
  Fourier transform one finds:
\begin{equation}
G(t,X;\tau, Y) = A \frac{\delta(t-\tau-D(X,Y))}{\sinh D(X,Y)}+B \frac{\delta(t-\tau+D(X,Y))}{\sinh D(X,Y)}
\end{equation}
The choices for $A$ and $B$ determine what combination of the advanced and retarded propagator we have. Note that this propagator is periodic with period $2 \pi i$ in the time coordinate.

\subsubsection{Li\'enard-Wiechert Potential}

The Li\'enard-Wiechert Potential describes the electromagnetic field
due to a charge $q$ moving in Minkowski space along some path $r(s)
\in \mathbb{E}^{3,1}$
where $s$ is any parameter. The potential at a point $x$ is
constructed as follows: first find a solution $s_r$ to the equation:
\begin{equation}
(x - r(s))\cdot(x - r(s)) = 0 \label{lightcone}
\end{equation}
which should correspond to the intersection of the path of the charge
with the \emph{past} light cone for a retarded propagator. The Maxwell
field is then determined by the one-form
\begin{equation}
A = \left. \frac{q}{4 \pi \epsilon_0} \frac{ \dot{r} \cdot dx}{(x
  - r)\cdot \dot{r}}\right |_{s=s_r}.
\end{equation}
Note that this is invariant under re-parameterisations of the path of
the particle $r(s)$. In calculating $F=dA$ one should be wary since
$A$ depends on $x$ both explicitly and also implicitly through
$s_r$. Differentiating (\ref{lightcone}) one finds that:
\begin{equation}
ds_r = - \frac{ (x - r) \cdot dx}{(x
  - r)\cdot\dot{r}}. 
\end{equation}
The standard calculations may now be performed and the Maxwell field
calculated.

In order to find the field  due to a point charge $q$
moving along the curve $(s,R(s)) \in \mathbb{R}_t\times \Ht$ we will
use the conformal invariance of the Maxwell equations. Using the
conformal map $\phi$ we may pull back the Li\'enard-Wiechert potential
from Minkowski space. One finds that the lightcone condition may be
re-written:
\begin{equation}
t-s_r = D(X, R(s_r)) \label{lightconehyp}
\end{equation}
and the potential is given by:
\begin{equation}
A =\frac{q}{4 \pi \epsilon_0} \left. \frac{1}{\dot{R} \cdot X - \sqrt{(R \cdot X)^2-1}} \left[(R\cdot
  X+\dot{R} \cdot X) dt + R \cdot dX + \dot{R} \cdot dX \right]\right |_{s=s_r}.
\end{equation}
Once again the dependence on $(X, t)$ is subtle, but one may calculate
the Maxwell field by using:
\begin{equation}
ds_r = \frac{\sqrt{(X\cdot R)^2 - 1}dt + R \cdot dX}{\sqrt{(X\cdot
    R)^2 - 1} - \dot{R}\cdot X},
\end{equation}
which follows from differentiating (\ref{lightconehyp}).

Calculating the field strength in the limit where the source charge is at a large distance from the observer but $\dot{R}\cdot\dot{R}$ and  $\ddot{R}\cdot\ddot{R}$ remain bounded, we find that the field decays like $e^{-\chi}$ with $\chi$ the separation of charge and observer.

We may interpret this in terms of  black hole optical geometry which approaches $\Ht \times \R_t$ near the horizon. In this case, as we shall see later one must add a static spherically symmetric field to enforce the condition that the black hole is uncharged. We have shown that even including the corrections to the electromagnetic field due to the motion of the charge, the field due to a particle falling into a black hole tends to a monopole charge as the particle approaches the horizon.

\subsubsection{Spinors on hyperbolic space}

Having dealt with spin 0 and spin 1 fields on $\mathbb{R}_t \times \Ht$, the logical next step is to discuss spinors and the Dirac operator on this space. Following Dirac \cite{Dirac}, we will write the Dirac equation in terms of objects in the embedding space, $\mathbb{E}^{3,1}$ as this will allow us to maintain $SO(3,1)$ covariance. 

We will make use of the following observation: the Dirac algebra for $\mathbb{E}^{3,1}$ may be represented in the form:
\begin{eqnarray}
\gamma^0 &=& i \sigma_3 \otimes I_2, \nonumber \\
\gamma^i &=& \sigma_2 \otimes \sigma_i,
\end{eqnarray}
where $\sigma_i$ are the Pauli matrices which form a 2-component representation of the Dirac algebra for $\mathbb{E}^3$. We could, if we so chose, construct this 2-component representation by considering Weyl spinors, $\chi^\alpha$ on an auxiliary $\mathbb{E}^{3,1}$ restricted to a constant time hyperplane, $\Sigma$. These 2-spinors would then transform under Poincar\'e transformations of this auxiliary $\mathbb{E}^{3,1}$ fixing the hyperplane and would have a natural $L^2$ inner product respecting these transformations given by
\begin{equation}
(\chi_1, \chi_2) = \int_\Sigma \mu_\Sigma \hspace{.1cm}\bar{\chi}_1  \bar{\sigma}\cdot T\chi_2 =  \int_\Sigma \mu_\Sigma \hspace{.1cm} \bar{\chi}_1 \chi_2\, .
\end{equation}
We use the notation $\sigma^\mu = (I_2, \sigma_i), \bar{\sigma}^\mu = (-I_2, \sigma_i)$. A Dirac spinor on the original $\mathbb{E}^{3,1}$ space is then the product of two 2-component spinors, the first transforming under $SO(1,1)$ which generates the boosts of the Lorentz group and the second under the rotational $SO(3)$. In fact the second spinors transform under the whole $E(3)$ symmetry of $\mathbb{E}^3$ but the translations act by the identity. The reason we take this somewhat circuitous approach to constructing the Dirac spinors is that it will allow us to construct spinors for $\mathbb{R}_t \times \Ht$ respecting the $SO(3,1)$ invariance of $\Ht$.

We will now make use of the fact that the metric in the forward light cone of the origin of Minkowski space may be written in the form
\begin{equation}
ds^2 = -dt^2 + t^2 h,
\end{equation}
with $h$ the metric on $\Ht$. This is sometimes referred to as the Milne universe. This Minkowski space will play the role of the auxiliary space above. Using the standard approach to construct the Dirac operator from the spin connection, one finds that acting on Weyl spinors,
\begin{equation}
\dir  = \bar{\sigma}^0 \left(\pd{}{t} + \frac{3}{2 t} \right) + \frac{1}{t} \dir^{(h)}. \label{4dir}
\end{equation}
Now any vector in the forward light cone of the origin may be written
\begin{equation}
W = t X, \qquad \textrm{with}\quad  X\cdot X = -1, \quad t>0, \quad X^0 >0
\end{equation}
we may re-write the Dirac operator as
\begin{equation}
\dir = \bar{\sigma} \cdot \pd{}{W} = \bar{\sigma}\cdot X \left(X\cdot \pd{}{W} + \frac{3}{2} \frac{1}{\abs{W}} \right) + \frac{1}{\abs{W}} \left(\abs{W} \bar{\sigma}\cdot \pd{}{W} - \bar{\sigma}\cdot X X \cdot \pd{}{W} - \frac{3}{2} \bar{\sigma}\cdot X  \right),
\end{equation}
comparing this with (\ref{4dir}) we conclude that
\begin{eqnarray}
\dir^{(h)} &=& \bar{\sigma}\cdot \nabla - \bar{\sigma}\cdot X X\cdot \nabla - \frac{3}{2} \bar{\sigma}\cdot X, \nonumber \\
&=& -\bar{\sigma}\cdot X \left(\frac{1}{2} \sigma^\mu \bar{\sigma}^\nu(X_\mu \nabla_\nu -  X_\nu \nabla_\mu) - \frac{3}{2}\right).
\end{eqnarray}
It may be checked that $\bar{\sigma}\cdot X$ anti-commutes with the RHS of this expression, so it is in fact convenient to take
{\begin{eqnarray}
\dir^{(h)} &=&\frac{1}{2} \sigma^\mu \bar{\sigma}^\nu(X_\mu \nabla_\nu -  X_\nu \nabla_\mu) - \frac{3}{2}\nonumber \\
&=& M-\frac{3}{2},
\end{eqnarray}
the operator $M$ is that introduced by Dirac in \cite{Dirac}. The relation to the standard construction for Dirac operators in a curved space is developed in \cite{Paramonov} Following the discussion above, there is a natural $L^2$ inner product on the space of Weyl spinors on $\Ht$ given by
\begin{equation}
(\chi_1, \chi_2)_\Ht = \int_\Ht \mu[X] \hspace{.1cm}\bar{\chi_1}  \bar{\sigma}\cdot X\chi_2 
\end{equation}
Here $\mu[X]$ is the Riemannian volume form of $\Ht$. This inner product is positive definite, because $X\cdot X=-1$ and $X^0>0$. Importantly, it is also $SO(3,1)$ invariant by construction and so respects all of the symmetries of hyperbolic space.

We can exhibit a set of plane wave eigenfunctions of the Dirac operator by considering the analogous plane wave eigenfunctions for the Laplace operator on $\Ht$ exhibited by Moschella and Schaeffer \cite{Moschella:2007zza}. Let $\chi$ be a constant 2-spinor (in the sense that $\nabla_\mu \chi = 0$), then the spinor
\begin{equation}
\kappa_{\omega \chi}(X) = \frac{\omega}{(2 \pi)^{\frac{3}{2}}} \left(\bar{\chi} \bar{\sigma}\cdot X \chi \right)^{-\frac{3}{2} - i \omega} \chi
\end{equation}
satisfies
\begin{equation}
\dir^{(h)} \kappa_{\omega \chi}(X) = i \omega \kappa_{\omega \chi}(X).
\end{equation}
Furthermore these functions tend pointwise to the standard plane wave basis for eigenfunctions on $\mathbb{E}^3$ as the radius of curvature of $\Ht$ tends to infinity. Obviously $\chi$ and $\lambda \chi$ define the same function, up to scale, for any $\lambda \in \mathbb{C}^*$ so that the space of eigenfunctions is $\mathbb{R}_+ \times \mathbb{CP}^1$. Using the results of Moschella and Schaeffer it should be possible to establish the following normalisation and completeness results
\begin{equation}
\left(\kappa_{\omega \chi}, \kappa_{\omega' \chi'} \right)_\Ht = \delta(\omega-\omega') \delta(\epsilon_{\alpha\beta}\chi^\alpha \chi'^\beta) 
\end{equation}
and
\begin{equation}
\int_{\mathbb{CP}^1} \mu[\chi] \int_0^\infty d\omega \hspace{.1cm} \kappa_{\omega \chi}(X) \overline{\kappa_{\omega \chi}}(X')\bar{\sigma}\cdot X' = \delta_\Ht(X, X') I_2 \label{comp}
\end{equation}
however this has so far proved difficult. We shall proceed therefore on the assumption that this is the case. For a discussion of integration over $\mathbb{CP}^1$ and the measure $\mu[\chi]$, see the appendix. One may at least show that the second result is true as the hyperbolic radius tends to infinity.

We can now construct Dirac spinors on $\mathbb{R}_t\times \Ht$ by taking the tensor product of a Weyl spinor on $\Ht$ with an $SO(1,1)$ spinor. The Dirac operator is given by:
\begin{equation}
\dir = i \sigma_3 \otimes I_2 \pd{}{t} + \sigma_2 \otimes \left(M-\frac{3}{2} \right).
\end{equation}
The Dirac conjugate is given by
\begin{equation}
\overline{\psi \otimes \chi} = \bar{\psi} i\sigma_3 \otimes \bar{\chi} \bar{\sigma}\cdot X
\end{equation}
we note that $\gamma^0 =  i \sigma_3 \otimes I_2$ so the Dirac conjugate does not take its standard form. This is related to the fact that we chose to make the Dirac operator on $\Ht$ more symmetric by multiplying by $\sigma \cdot X$. We note finally that we may identify
\begin{equation}
\gamma^5 = \sigma_1 \otimes I_2
\end{equation}
as the chirality matrix which satisfies
\begin{equation}
\db{\gamma^5}{\dir}=0 \qquad \textrm{and} \qquad \left(\gamma^5\right)^2 = -I_4.
\end{equation}
We may finally construct a complete set of plane wave solutions to the Dirac equation
\begin{equation}
\dir \Psi = 0
\end{equation}
as follows:
\begin{equation}
\Psi^0_{s \omega z \lambda}(X,t) = \frac{\omega}{(2 \pi)^\frac{3}{2}} e^{-si\omega t} \psi_{\lambda} \otimes \left(\bar{\chi}_z \bar{\sigma}\cdot X \chi_z  \right)^{-\frac{3}{2}+s \lambda i \omega} \chi_z 
\end{equation}
where
\begin{equation}
\psi_\lambda = \frac{1}{\sqrt{2}}\binom{1}{\lambda}, \qquad \textrm{and} \qquad 
\chi_z = \frac{1}{\sqrt{1+\abs{z}^2}} \binom{1}{ z}\, .
\end{equation}
We have introduced the quantum numbers $s=\pm$ which distinguishes the positive and negative energy solutions, $z\in \mathbb{C}$ which parameterizes $\mathbb{CP}^1$ under the usual stereographic projection and $\lambda = \pm$, the chirality. With this choice of parameterisation of $\mathbb{CP}^1$ the appropriate measure in the completeness relation (\ref{comp}) is
\begin{equation}
\mu[\chi] = \frac{2 i dz d\bar{z}}{\left(1+\abs{z}^2 \right)^2}
\end{equation}
which we recognise as the measure on $S^2$ under stereographic projection.

We will now use these results to calculate the force between electrons (or indeed other leptons) mediated by neutrino exchange.

\subsubsection{Neutrino mediated forces \label{neutf}}

In flat space there is a long range lepton-lepton force mediated by the exchange of a pair of neutrinos. The potential, as shown by Feinberg and Sucher \cite{Feinberg}, is 
\begin{equation}
V(r) = \frac{G_W^2}{4 \pi^3 r^5}\, ,
\end{equation}
where $G_W$ is the weak-interaction coupling constant. Their calculation was based on calculating the one-loop scattering of one electron by another mediated by a $\nu \bar{\nu}$ pair. A simpler means of finding this potential, as described by Hartle \cite{Hartle:1972jj}, is to treat the neutrino field as quantum mechanical and the electrons as classical both in their role as a source for the neutrino field and as particles acted on by that field. Hartle shows that in this limit, the neutrino field obeys the modified Dirac equation
\begin{equation}
\left(i\dir -\frac{G_W}{\sqrt{2}} \gamma \cdot N (1+\gamma^5)\right) \Psi(x) = 0 \label{direqn}
\end{equation}
where $N_\mu$ is the classical electron number current. The equation of motion of an electron in the classical limit is
\begin{equation}
m \frac{D u^\mu}{d\tau} = \frac{G_W}{\sqrt{2}} u^\nu \left(\partial^\mu B_\nu - \partial_\nu B^\mu \right)
\end{equation}
where $u^\mu$ is the electron's four-velocity and the potential $B^\mu$ is given in terms of the neutrino field by:
\begin{equation}
B^\mu(x) = \left \langle \bar{\Psi}(x) \gamma^\mu (1+\gamma^5)\Psi(x) \right \rangle -  \left \langle \bar{\Psi}^0(x) \gamma^\mu (1+\gamma^5)\Psi^0(x) \right \rangle.
\end{equation}
$\Psi(x)$ is the neutrino field with the weak interactions turned on and $\Psi^0(x)$ is the same field with the interactions turned off. Both expectation values are taken in vacuum with no free neutrinos. The normalisation of the neutrino field is fixed by the canonical anti-commutation relations which relate the anti-commutators of fields on a spacelike hypersurface $\sigma$. The only non-vanishing bracket is
\begin{equation}
\db{\Psi(x)}{\bar{\Psi}(x') \gamma\cdot T} = \delta_\sigma(x, x') \label{cacr}
\end{equation}
where $T$ is the future directed unit normal to $\sigma$. Note that the standard relation would be between $\Psi$ and $\Psi^\dagger$, however as noted above, we have chosen to make the Dirac operator simpler at the expense of taking a non-standard Dirac conjugate. For flat space sliced along constant $t$ hyperplanes, (\ref{cacr}) reduces to the standard relation.

Let us suppose that there is a complete set of solutions to the modified Dirac equation (\ref{direqn}) with the same quantum numbers as for the source free Dirac equation, so that we may write
\begin{equation}
\Psi_{s \omega z \lambda}(t, X) = e^{-s i \omega t} \psi_\lambda \otimes \kappa_{s \omega z \lambda}(X)
\end{equation}
and we will assume the completeness relation
\begin{equation}
\sum_s \int_\mathbb{C} \frac{2 i dz d\bar{z}}{\left( 1+\abs{z}^2\right)^2} \int_0^\infty  d\omega \hspace{.1cm}\kappa_{s \omega z \lambda}(X) \bar{\kappa}_{s \omega z \lambda}(X') \bar{\sigma}\cdot X' = \delta_\Ht (X, X') I_2
\end{equation}
(note that we don't sum over $\lambda$ here). We may therefore expand the neutrino field in the form
\begin{equation}
\Psi(t,X) = \sum_\lambda \int_\mathbb{C} \frac{2 i dz d\bar{z}}{\left( 1+\abs{z}^2\right)^2} \int_0^\infty  d\omega \hspace{.1cm} \left \{ e^{-i \omega t} \psi_\lambda \otimes \kappa_{+\omega z \lambda} b_{\omega z \lambda} + e^{i \omega t} \psi_\lambda \otimes \kappa_{-\omega z \lambda} d^\dagger_{\omega z \lambda} \right \}.
\end{equation}
The canonical anti-commutation relations for the neutrino field imply the following non-vanishing relations for the creation operators $b$ and $d$
\begin{equation}
\db{b_{\omega z \lambda}}{b^\dagger_{\omega' z' \lambda'}} = \db{d_{\omega z \lambda}}{d^\dagger_{\omega' z' \lambda'}} = \frac{\left(1+\abs{z}^2\right)^2}{2 i} \delta(z-z') \delta(\omega-\omega') \delta_{\lambda \lambda'}
\end{equation}
while all other brackets vanish. We suppose the existence of a vacuum state $\ket{0}$ such that $b_{\omega z \lambda}\ket{0}=d_{\omega z \lambda}\ket{0} = 0$. Using the anti-commutation relations we find that the electric part of the neutrino mediated vector potential, $V = B^0$ takes the form
\begin{eqnarray}
B^0(X) 
&=&-2  \int_\mathbb{C} \frac{2i dz d\bar{z}}{\left( 1+\abs{z}^2\right)^2} \int_0^\infty  d\omega \hspace{.1cm} \bigl \{\bar{\kappa}_{-\omega z +}(X) \bar{\sigma}\cdot X \kappa_{-\omega z +}(X)  \nonumber \\
&&\left . \qquad -\bar{\kappa}^0_{-\omega z +}(X) \bar{\sigma}\cdot X \kappa^0_{-\omega z +}(X)\right\}.
\end{eqnarray}
The fact that the coupling has a $1+\gamma^5$ factor ensures that only positive chirality modes contribute.

As we are only interested in effects at the lowest order in $G_W$, we may consider an expansion of the spinors $\kappa$ in terms of $G_W$. Expanding to first order
\begin{equation}
\kappa_{s \omega z\lambda} = \kappa^0_{s \omega z\lambda} +G_W \kappa^1_{s \omega z\lambda} 
\end{equation}
we find that provided we assume that $N_t = \delta_\Ht(X, X')$ is the only non-zero component of the electron current, the modified Dirac equation (\ref{direqn}) implies that
\begin{eqnarray}
\left(M - {\textstyle \frac{3}{2}} + i s \lambda \omega\right) \kappa^0_{s \omega z \lambda} &=& 0 \nonumber \\
\left(M - {\textstyle \frac{3}{2}} + i s \lambda \omega\right) \kappa^1_{s \omega z \lambda} &=& i\sqrt{2}\delta_\Ht(X,X') \kappa^0_{s \omega z \lambda}. \label{gexp}
\end{eqnarray}
We know from the previous section that the normalised zero'th order spinors take the form
\begin{equation}
 \kappa^0_{s \omega z \lambda}= \frac{\omega}{(2 \pi)^\frac{3}{2}} \left(\bar{\chi}_z \bar{\sigma}\cdot X \chi_z  \right)^{-\frac{3}{2}+i s \lambda \omega} \chi_z\, . \label{kap0}
\end{equation}
In order to solve the second equation we note that
\begin{equation}
M^2 -2 M = \nabla^2_\Ht\, ,
\end{equation}
where the Laplacian here is the \emph{scalar} Laplacian acting componentwise to the right. This may be verified by following the argument of Dirac \cite{Dirac} with the appropriate signature and dimension. We see that 
\begin{equation}
\left(M-{\textstyle \frac{3}{2}} + i s \lambda \omega \right) \left (M-{\textstyle \frac{1}{2}}-i s \lambda \omega \right) = \left[ \nabla_\Ht^2 + 1 +\left (\omega  + \frac{i s \lambda}{2}\right)^2\right ] I_2
\end{equation}
we recognise the right hand side of the equation as the conformal wave equation with a complex wavenumber, for which we have already found the Green's function. Thus we may solve the second of equations (\ref{gexp}) by
\begin{equation}
 \kappa^1_{s \omega z \lambda}(X) = \left [ \left( M-{\textstyle \frac{1}{2}}-i s \lambda \omega\right) \phi(X,X') \right] \kappa^0_{s \omega z \lambda}(X') ,
\end{equation}
where
\begin{equation}
\phi(X,X') = \frac{i \sqrt{2}}{4\pi} e^{i s \omega \chi} \frac{e^{ - \lambda \chi /2}}{\sinh \chi}, \qquad \textrm{with}\qquad \cosh \chi = - X \cdot X'. \label{scalgrn}
\end{equation}
There is a choice of sign here corresponding to picking the retarded propagator. We note that after Fourier transforming back to the time domain the propagator will be anti-periodic in time with period $2 \pi i$. We also note that there appears to be a breaking of the symmetry one might expect under $\lambda \to -\lambda$. This chiral symmetry breaking is a subtle consequence of the negative curvature and is discussed in \cite{Callan, Gorbar}.

Putting this all together, we find
\begin{equation}
\frac{B^0(X)}{G_W} = -4 \textrm{Re} \int_\mathbb{C} \frac{2 i dz d\bar{z}}{\left( 1+\abs{z}^2\right)^2} \int_0^\infty  d\omega \hspace{.1cm} \bar{\kappa}^0_{-\omega z +}(X)\bar{\sigma}\cdot X \left(M-{\textstyle \frac{1}{2}}+i \omega  \right) \phi(\chi) \kappa^0_{- \omega z +}(X').
\end{equation}
Inserting (\ref{kap0}) for $\kappa^0$, the integrand reduces to the form
\begin{eqnarray}
&& \frac{\omega^2}{(2 \pi)^3} \biggl [ (\bar{\chi}_z \bar{\sigma}\cdot X \chi_z)^{-\frac{3}{2}+ i \omega} (\bar{\chi}_z \bar{\sigma}\cdot X' \chi_z)^{-\frac{1}{2}- i \omega} \frac{\phi'(\chi)}{\sinh \chi} \nonumber \\ &&
\qquad +  (\bar{\chi}_z \bar{\sigma}\cdot X \chi_z)^{-\frac{1}{2}+ i \omega} (\bar{\chi}_z \bar{\sigma}\cdot X' \chi_z)^{-\frac{3}{2}- i \omega} \bigl( (i \omega - {\textstyle \frac{1}{2}}) \phi(\chi) - \coth \chi \phi'(\chi) \bigr ) \biggr ]. \label{integr1}
\end{eqnarray}
We will first perform the integrals over $\mathbb{CP}^1$ which are of the form:
\begin{equation}
I_a=\int_\mathbb{C}  \frac{2 i dz d\bar{z}}{\left( 1+\abs{z}^2\right)^2}  (\bar{\chi}_z \bar{\sigma}\cdot X \chi_z)^{-a} (\bar{\chi}_z \bar{\sigma}\cdot X' \chi_z)^{a-2}. \label{neutint}
\end{equation}
One may verify that
\begin{equation}
(\bar{\chi}_z \bar{\sigma}^\mu \chi_z) = \left(1, -\vect{n}_z \right),
\end{equation}
where $\vect{n}_z$ is the pull back of $z$ to the unit sphere in $\mathbb{R}^3$ under the standard stereographic projection map. $I_a$ is Lorentz invariant (see Appendix), so we may assume without loss of generality that
\begin{eqnarray}
X &=& (\cosh \chi, 0, 0, -\sinh \chi), \nonumber \\
X' &=& (1, 0, 0, 0),
\end{eqnarray}
and we may integrate over $S^2$ using standard spherical polar coordinates so that
\begin{equation}
I_a = \int \frac{\sin \theta d \theta d\phi}{(\cosh\chi + \cos \theta \sinh \chi)^a} = \frac{4 \pi \sinh (1-a) \chi}{(1-a) \sinh \chi}.
\end{equation}
Note that this is symmetric under $a \to 2-a$ as it must be since we could have chosen $X$ and $X'$ the other way around. Integrating (\ref{integr1}) over $\mathbb{CP}^1$ then, we have after some simplification
\begin{equation}
\frac{i \sqrt{2} \omega^2}{8 \pi^3} \left( \frac{1}{\sinh^2 \chi} + \frac{2}{(1+4 \omega^2)\sinh^4\chi}\right) + \frac{\sqrt{2} \omega^2}{8\pi^3(1+4 \omega^2)} \frac{e^{-2 i \chi \omega}}{\sinh^4\chi}\left( 2 \omega \sinh \chi - i \cosh \chi \right),
\end{equation}
so that
\begin{equation}
\frac{B^0(\chi)}{G_W} = -\frac{4 \sqrt{2}}{4 \pi^3} \textrm{Re} \int_0^\infty d\omega \left \{ \frac{\omega^2}{1+4 \omega^2} \frac{e^{-2 i \chi \omega}}{\sinh^4\chi}\left( 2 \omega \sinh \chi - i \cosh \chi \right)\right \}.
\end{equation}
This integral is manifestly divergent for large $\omega$, however the potential we are interested in is a low energy effect, so we may introduce a large momentum cut-off by sending $\chi \to \chi - i \epsilon$ which will make the integral converge for large $\omega$ and taking the $\epsilon \to 0$ limit after calculating the integral. Doing this, we find that the integral may be performed exactly and we find that the neutrino field gives rise to an effective potential:
\begin{equation}
V(\chi) = \frac{G_W}{\sqrt{2}} B^t(\chi) = \frac{G_W^2}{8 \pi^3 \chi^2 \sinh^4\chi}\left(\chi \cosh \chi + \sinh \chi - \chi^2\, \textrm{Shi} \hspace{.09cm}\chi \right),
\end{equation}
where $\textrm{Shi}$ is the sinh integral:
\begin{equation}
\textrm{Shi} \hspace{.09cm} \chi = \int_0^\chi \frac{\sinh t}{t} dt\, .
\end{equation}
This equation is valid for any $\chi$ large with respect to the length scale defined by $G_W$. We may take the limit where the hyperbolic radius of the space tends to infinity and we find that
\begin{equation}
V(\chi) \sim \frac{G_W^2}{4 \pi ^3 \chi^5}
\end{equation}
which is the result of Feinberg and Sucher for flat space.

One might be concerned by the fact that there appears to be an asymmetry between the right and left handed neutrinos implicit in (\ref{scalgrn}) however it is possible to perform the same calculation under the assumption that left handed neutrinos couple to electrons and the answer found is precisely the same.

\subsection{Thermodynamics}

We noted above that the scalar propagator was periodic and the fermion propagator anti-periodic in imaginary time. Reinserting dimensions, the period is given by $2 \pi i R$ where $R$ is the radius of the hyperbolic space. One expects that thermal propagators for a field at temperature $T$ should have an imaginary period equal to $1/T$, the inverse of the temperature. As we remarked above, for the near horizon optical geometry of a horizon with surface gravity $\kappa$, $R=\kappa^{-1}$ so we find that fields in the neighbourhood of a horizon are thermalized at a temperature
\begin{equation}
T_H = \frac{\kappa}{2\pi}\, ,
\end{equation}
precisely the Hawking temperature of the horizon. Notice that we nowhere had to `Euclideanize' the time direction of the manifold in order to derive this result.

\subsection{Approximate calculations near the horizon}

Let us consider trying to construct the propagator for some physical
field in a black hole background in the limit where the pole of the
propagator approaches the horizon. We assume that since we are
considering a small perturbation to the background that the equations
are linear and moreover that they may be converted by a conformal
transformation to equations with respect to the optical
metric. Further, since the black-hole background is assumed to be
static, we may Fourier transform with respect to $t$ so that the
equations may be expressed in terms of the optical geometry of
$t=\mathrm{const.}$ slices. We divide these slices into two regions:
\begin{itemize}
\item {\bf Region I} is a region surrounding the horizon, such that
  in this region the optical metric has constant negative curvature to order
  $\epsilon$.
\item {\bf Region II} is the complement of region I and contains the
  asymptotically flat end.
\end{itemize}
Region I may be thought of as the exterior of a ball in $\Ht$ in the case where the horizon has spherical topology. One would expect that if the topology differs for that of the sphere, then this will not have an effect on the propagator in the limit where the pole approaches the horizon, so we assume that Region I is indeed of this form. The
conformal infinity of the hyperbolic space corresponds to the horizon
of the black hole. We wish to solve the problem:
\begin{equation}
\mathcal{L}\phi(x) = \mathcal{L}_{h_{opt.}} \phi(x) + \mathcal{L}' \phi(x) = \delta(x,x_0)\, ,
\end{equation}
where $x_0$ is close to the horizon and we have split the linear
operator $\mathcal{L}$ into a geometric operator constructed from
$h_{opt.}$ and another part $\mathcal{L}'$ which is assumed to be
small in Region I. This may require shrinking Region I. We do not assume here that $\phi$ is a scalar -- the same
considerations will apply for fields of any spin.

In region I, the problem simplifies to finding a propagator in
hyperbolic space. This is a simplification because $\Ht$ is maximally
symmetric so one may make use of isometries to move the pole of the
propagator around. We define $G_I(x,x_0)$ to satisfy the equation on
hyperbolic space:
\begin{equation}
\mathcal{L}_h G_I(x, x_0) = \delta(x, x_0)
\end{equation}
subject to suitable boundary conditions as $x$ approaches conformal
infinity (i.e.\ the horizon). 

In region II, we are solving the homogeneous problem
\begin{equation}
\mathcal{L}_{h_{opt.}} G_{II}(x,x_0) + \mathcal{L}' G_{II}(x,x_0) = 0
\end{equation}
such that $G_{II}(x, x_0)$ agrees with $G_{I}(x,x_0)$ at the boundary
between regions I and II and decays suitably as $x$ approaches the
asymptotically flat infinity. The approximate propagator we construct
is then given by:
\begin{equation}
G(x, x_0) =\left\{ \begin{array}{ll}
 G_I(x,x_0) + K(x) & \textrm{if $x$ in Region I}\\
 G_{II}(x,X_0) + K(x) & \textrm{if $x$ in Region II}\\
  \end{array} \right. ,
\end{equation}
$K$ here is any solution of the homogeneous problem on the whole of
the exterior of the black hole which satisfies appropriate boundary
conditions both at spacelike infinity and at the horizon. It is these
solutions which carry any `hair' which the black hole may have. The
charges carried by the black hole as a result of $K$ do not follow
from regularity at infinity or the horizon but must be determined by,
for example, integral conservation laws.

We are now ready to describe the limit as the pole of the propagator
approaches the horizon. In Region I we see this as the pole of a
propagator in $\Ht$ moving towards conformal infinity. As the boundary
of Region I is a fixed compact surface in $\Ht$, the fields on the
boundary decay. Typically this decay is exponentially quickly in the
hyperbolic distance of the pole from some fixed point. Thus $G_{II}$
will also decay at this rate, by linearity. Thus only $K$ can remain
in the limit as the pole approaches the horizon, with all other terms
decaying. If the black hole cannot support a regular external field
$K$, then the fields must all approach zero as the pole of the
propagator approaches the horizon. In the case of spherical symmetry it is useful to take the
boundary of Region I to be a sphere as it its then possible to
decompose all the functions into spherical harmonics and the decay
rates for each multipole moment can be calculated separately.

We thus have a method to calculate the rates of decay of propagators
as their poles approach the horizon in terms of the distance in the
optical metric. We may relate the optical distance along a radial
geodesic starting from some fixed point, $\chi$, to the proper distance to the
horizon along that geodesic in the physical metric, $\delta$, by:
\begin{equation}
\delta \sim C e^{-\chi}
\end{equation}
in the region near the horizon. This allows us to re-express the decay
rates in terms of proper distance to the horizon in the physical
metric. We will give constructions below for some simple propagators
in hyperbolic space which are useful when constructing the approximate
propagators in region I.

\subsubsection{Example - Massive scalar field}

In the case of a massive scalar field satisfying the Klein-Gordon
equation, we wish to find a propagator which satisfies:
\begin{equation}
\Box_g \psi - m^2 \psi = \delta(x,x_0).
\end{equation}
Using conformal transformations defined in section (\ref{scalfield}),
solving this is equivalent to solving the equation
\begin{equation}
\frac{1}{H^5}\left(\Delta_h \Phi + \Phi\right) + \frac{1}{H}\left(k^2
- m^2 \Omega^2 \right)\Phi = \frac{1}{H^6} \delta_h(x,x_0),
\end{equation}
where $h$ refers, as always, to the metric on $\Ht$ and $H\to1$,
$\Omega \to 0$ as we approach the horizon. This is of the form
supposed above and the exterior supports no solutions to the
homogeneous Klein-Gordon equation except the zero solution. The
propagator in region I is given by:
\begin{equation}
\frac{e^{i k \chi}}{\sinh \chi}, \qquad \mathrm{where}\quad  \chi = D(x,x_0),
\end{equation}
with $D(p,q)$ the distance in the optical metric between $p$ and
$q$. Thus as $x_0$ goes to conformal infinity, the propagator falls
off like $e^{-\chi}$. In terms of the proper distance of the pole of the
propagator from the horizon $\delta$, we find that the propagator
vanishes like $\delta^1$. This is in agreement with Teitelboim \cite{Teitelboim:1972qx}. We will verify this analysis below by showing that for the $k=0$ mode
we may solve the problem exactly throughout the exterior.

\subsubsection{Example - Proca equation}

The generalisation of Maxwell's equations for electromagnetism to the case where the photon is not taken to be massless is given by the Proca equations. In terms of a one-form $A$ the vacuum equations may be written:
\begin{equation}
\star d \star d A + m^2 A = 0\, .
\end{equation}
In this case, we may quickly estimate the rate at which the information is lost as a charged particle falls quasi-statically into a (Schwarzschild or Reissner-Nordstr\"om) back hole if the electromagnetic field is mediated by a massive vector boson. We make the ansatz:
\begin{equation}
A = \frac{\psi(\vect{x})}{H(\vect{x})} dt,
\end{equation}
and find that for a point charge at $x_0$ the function $\psi$ should satisfy:
\begin{equation}
\frac{1}{H^5} \Delta_{h} \psi - m^2 \frac{\Omega^2}{H} \psi = \frac{1}{H^6} \delta_{h}(x, x_0),
\end{equation}
where $h, H$ and $\Omega$ are as in the last section. This is once again of the form conjectured above and for $m \neq 0$ there are no solutions to the vacuum equations regular throughout the exterior of the black hole. The $m=0$ case corresponds to a Maxwell field and is treated below, however we expect a significant difference as for this case the equations are gauge invariant. 

The propagator in region I is given by:
\begin{equation}
\frac{1}{e^{2 \chi}-1}, \qquad \mathrm{where}\quad  \chi = D(x,x_0),
\end{equation}
where $D(p, q)$ is as above. As the pole recedes to infinity the propagator decays like $e^{-2 \chi}$ corresponding to a fall off as the square of the proper distance to the horizon, $\delta^2$. This decay at twice the rate of the massive scalar boson case is also in agreement with Teitelboim.

\subsubsection{Example - Forces from Neutrino pair exchange}

As noted above there is a force in flat space between leptons mediated by neutrinos which might in principle be used to measure the lepton number of a black hole. As a black hole should not have a measurable lepton number associated with it, we will now consider the problem of neutrino mediated forces in the vicinity of an event horizon. We expect such forces to vanish as a lepton approaches the horizon. We will require the following short Lemma which may be proven by considering the behaviour of the spin connection under conformal transformations.

\begin{lemma}
Suppose $\tilde{g} = \Omega^2 g$ are two conformally related $n$-dimensional metrics and \makebox{$\tilde{\psi} = \Omega^{-\left( n-1\right)/2} \psi$} is a Dirac spinor, then
\begin{equation}
\widetilde{\not\hspace{-.09cm} D} \tilde{\psi} =   \Omega^{-\left(n+1\right)/2}  \not\hspace{-.09cm} D\psi.
\end{equation}
\end{lemma}

Thus solutions of Dirac's equation for $g$ may be rescaled to solutions for $\tilde{g}$. Also
\begin{equation}
\int_\sigma d\tilde{\sigma} \bar{\tilde{\phi}}\tilde{\psi} = \int_\sigma (d\sigma \Omega^{n-1})(\Omega^{\frac{1-n}{2}})^2\bar{\phi}\psi = \int_\sigma d\sigma \bar{\phi}\psi,
\end{equation}
so that orthonormality is preserved. It is not true however that if we start with a complete set then after rescaling we have a complete set. We will now restrict to the case where $d=4$ and the conformal transformation depends only on the spatial coordinates $X$ so that the metric remains static.

By considering how the equation (\ref{direqn}) transforms under such a conformal transformation then \emph{assuming the rescaled solutions to Dirac's equation are complete} we may calculate the interaction potential for the neutrino mediated force. Suppose that $V(X,X')$ represents the potential at $X$, due to an electron at $X'$, with the metric $g$. Then the potential $\tilde{V}(X,X')$ for the metric $\tilde{g}$ is given by:
\begin{equation}
\tilde{V}(X,X') = \Omega^{-3}(X) \Omega^{-2}(X') V(X, X').
\end{equation}
There is clearly an asymmetry between the source and the test particle, however this is due the the redshift effect which means that an energy measured at different spatial points will vary.

As an example, we may consider conformally rescaling $\mathbb{R}_t\times\Ht$ to a metric on a static patch of the de Sitter space. In this case we pick an arbitrary point $X^0$ and we may write the de Sitter metric as
\begin{equation}
\tilde{g} = \frac{1}{(-X \cdot X^0)^2} g,
\end{equation}
where $g$ is the metric on $\mathbb{R}_t\times\Ht$. Suppose we take a patch with the observer at the origin, then the potential measured by this observer due to an electron at $X'$ is given by:
\begin{equation}
V(\chi) = \frac{G_W^2 \cosh^2 \chi}{8 \pi^3 \chi^2 \sinh^4\chi}\left(\chi \cosh \chi + \sinh \chi - \chi^2\, \textrm{Shi} \hspace{.09cm}\chi \right),
\end{equation}
with $\cosh\chi = -X\cdot X'$. As the electron approaches the horizon, $\chi \to \infty$ and the potential is extinguished like $e^{-\chi}/\chi^3$, thus demonstrating the no-hair property of the de Sitter cosmological horizon for neutrino mediated forces.

Unfortunately this method does not work completely for the Schwarzschild event horizon because the conformally rescaled solutions of the Dirac equation do not form a complete basis, essentially because neutrinos may start either at spatial infinity or at the horizon, this point is made by Teitelboim and Hartle \cite{Teitelboim:1972qx, Hartle:1972jj}. Accordingly, we do not reproduce precisely the extinction rate of Hartle, who finds the potential vanishes like $e^{-\chi}/\chi$, but we do find the correct exponential rate.

\subsection{Exact Calculations for Schwarzschild and Reissner-Nordstr\"om}

\subsubsection{Electric Charge}

We would like to find the field due to a static electric charge in the
Schwarzschild or Reissner-Nordstr\"om background. In order to do this, we make use of the fact
that Maxwell's equations in vacuo:
\begin{equation}
\label{maxwell} d F =0, \qquad d \star_{g} F = 0,
\end{equation}
are conformally invariant in $4$ space-time dimensions. Thus if we
can find the field due to a point particle with respect to the optical
metric $\gopt$ then we know the field with respect to the physical
metric. The reason for this is that in $4$ spacetime dimensions, for any
$2$-form $\omega$ with conformal weight $0$, we have:
\begin{equation}
\star_{\Omega^2 g} \omega = \star_{g} \omega. \label{2conf}
\end{equation}
Thus the Maxwell action
\begin{equation}
S = \int_\mathcal{M} F \wedge \star F
\end{equation}
does not change under a conformal transformation. It is also helpful
to note at this stage that the charge contained inside a $2$-surface
$\Sigma$:
\begin{equation}
Q_\Sigma = \int_\Sigma \iota^* (\star F )
\end{equation}
is conformally invariant, where $\iota:\Sigma \to \mathcal{M}$ is the
inclusion map and $\star$ may refer to any representative of the
conformal class of $g$, by (\ref{2conf}).

In order to solve the first Maxwell equation, we as usual introduce a
one-form potential $A$ and make a static ansatz:
\begin{equation}
F = dA, \qquad A = \Phi(\vect{x}) dt\, .
\end{equation}
The reason for this
static ansatz is primarily the fact that the resulting equations are
analytically tractable. It will give a good approximation to the field
of a freely falling particle, provided the particle is not moving very
quickly. Alternatively, one may imagine a thought experiment where a
charge is lowered from infinity towards the black hole horizon and
measurements of the fields are made as the charge approaches the black
hole.

The second Maxwell equation becomes the familiar Laplace equation for
$\Phi$ with respect to the optical metric:
\begin{equation}
d \star_{\hopt} d \Phi = 0.
\end{equation}
We will make use of the fact noted above that
\begin{equation}
\hopt = H^{4} h,
\end{equation}
where $h$ is the hyperbolic metric of $\Ht$ with radius $1$ and $H$ is
harmonic on $\Ht$. We are therefore able to relate the Laplacian of
$\hopt$ to the Laplacian of $h$. We will require the following lemma:
\begin{lemma}
 \label{lem} Suppose $h_1$ and $h_2$ are two three dimensional metrics with related
Laplace operators $^1\Delta$ and $^2\Delta$. Further suppose that they
are conformally related:
\begin{equation}
h_1 = H^4 h_2.
\end{equation}
Then if $\phi = H^{-1} \Psi$ the following relation holds:
\begin{equation}
^1\Delta \phi = \frac{1}{H^5} {}^2\Delta \Psi - \frac{\Psi}{H^6}
  {}^2\Delta H. \label{lap}
\end{equation}
\end{lemma}
\begin{proof}
Use the standard formula $\Delta = \frac{1}{\sqrt{g}} \pd{}{x^i}
\sqrt{g} g^{ij} \pd{}{x^j}$ and make the substitutions above, then
collect terms to find (\ref{lap}).
\end{proof}

We wish to find the Green's function for the Laplacian on
$\tilde{h}=\hopt$. This satisfies the following equation:
\begin{equation}
\widetilde{\Delta} \widetilde{G}(x, x_0) = \delta_{(\tilde{h})}(x,
x_0), \label{green}
\end{equation}
Where $\widetilde{\Delta}$ is taken to act on the $x$ coordinates. The
Dirac delta function is defined by the requirement:
\begin{equation}
\int_\mathcal{U} \delta_{(g)}(x, x_0) d \mathrm{vol_g} = \left\{ \begin{array}{ll}
 1 & \textrm{if $x_0 \in \mathcal{U}$}\\
 0 & \textrm{if $x_0 \notin \mathcal{U}$}\\
  \end{array} \right.
\label{rhon}
\end{equation}
for an open subset $\mathcal{U} \subset \mathcal{M}$.
Applying the above lemma and making use of the fact that $\Delta_h
H=0$ we have that
\begin{equation}
\widetilde{\Delta} \phi = \frac{1}{H^5} \Delta_h \Psi,
\end{equation}
with $\phi$ and $\Psi$ related as above. Inserting this into the
definition of the Green's function we have $\widetilde{G}(x, x_0) =
H(x)^{-1} G(x,x_0)$, where $G$ obeys:
\begin{equation}
\Delta_h G(x,x_0) = H(x)^5 \delta_{(\tilde{h})}(x,x_0) = H(x_0)^{-1}
\delta_{(h)}(x,x_0), \label{G}
\end{equation}
where in this last step we use properties of the Dirac delta
function. It would appear that we have simply replaced one Green's
function problem with another, however the great advantage is that we
now seek Green's functions on $\Ht$ which is maximally symmetric so if
we can find the Green's function for $x_0=0$ we can generate all
Green's functions by $SO(3,1)$ transformations.

Suppose $V_O(x)$ satisfies
\begin{equation}
\Delta_h V_O(x) = \delta_{(h)}(x, O), \qquad V_O \to \frac{1}{4
  \pi}\quad  \textrm{as} \quad D(x, O) \to \infty,
\end{equation}
where $O$ is a fixed point in $\Ht$ and $D(x,O)$ is the hyperbolic
distance from $x$ to $O$. In Beltrami coordinates $V_O(\vect{x})= (4
\pi \abs{\vect{x}})^{-1}$. By $SO(3,1)$ invariance,
for any other point $x_0$ we can find a isometry
$T:\Ht \to \Ht$ which satisfies
\begin{equation}
T_{x_0}^*h = h, \qquad T_{x_0}(x_0)=O. \label{Tdef}
\end{equation}
We then define
\begin{equation}
V_{x_0} = T_{x_0}^* V_O, \qquad \mathrm{i.e.} \quad V_{x_0}(x) = V_O(T_{x_0}(x)).
\end{equation}
The map $T$ is not uniquely defined, but since $V_O$ is
spherically symmetric any two maps $T$ satisfying (\ref{Tdef}) give
the same function $V_{x_0}$. This new function satisfies
\begin{equation}
\Delta_h V_{x_0}(x) = \delta_{(h)}(x, x_0).\label{Vx0}
\end{equation}
Putting together (\ref{G}) and (\ref{Vx0}) we find that the Green's
function for the Laplacian of $\hopt$ has the form
\begin{equation}
\widetilde{G}(x,x_0) = \frac{1}{H(x) H(x_0)}\left( V_O(T_{x_0}(x))+A\right) + B.
\end{equation}
We have now to specify boundary conditions. The constant $B$ is
unphysical, and if we choose $B=0$, the potential will vanish at the
asymptotically flat end. The constant $A$ arises because if $V_{x_0}$
satisfies (\ref{Vx0}) then so does $V_{x_0}+A$. This constant gives
rise to a non-trivial field which corresponds to the black hole
carrying a (linearised) charge. Transforming to Eddington-Finkelstein
coordinates shows that the function $\widetilde{G}$ is regular at the
horizon for all values of $A$, so we must look elsewhere for our final
boundary condition. This comes from the fact that Gauss' law should be
satisfied. If one considers a surface $\Sigma$ which encloses the
black hole, but not the point $x_0$ then $Q_\Sigma$ should
vanish. Enforcing this condition fixes $A$. In the case where the harmonic function $H$
takes the Reissner-Nordstr\"om form:
\begin{equation}
\sqrt{\mu} H(x) = 4 \pi \mu V_O(x) + m
\end{equation}
Gauss' law requires $A = \frac{m}{4 \pi \mu}$ and we finally have:
\begin{equation}
\widetilde{G}(x,x_0) = \frac{1}{H(x) H(x_0)}\left(
V_O(T_{x_0}(x))+\frac{m}{4 \pi \mu}\right). \label{fulG}
\end{equation}
This construction is valid for both Schwarzschild and
Reissner-Nordstr\"om and gives the linear perturbation to the
electromagnetic field due to a static point charge located in the
spacetime. The potential has been found in terms of geometric objects
of hyperbolic space, and so is valid for any coordinate system on $\Ht$.

We can see from this equation how the information associated with the
precise location of the charged particle is lost as it is lowered
towards the black hole. The only term which is not spherically
symmetric about $O$ in (\ref{fulG}) is the $V_O$ term. As the point
$x_0$ recedes from $O$ towards the black hole horizon which is at the
conformal infinity of $\Ht$, this term approaches a constant
exponentially quickly in $D(x_0, O)$. The potential tends to
the spherically symmetric field associated with the black hole
carrying a charge and deviations from this field fall exponentially
with $D(x_0, O)$.

We plot below the isopotentials for a point charge in
Schwarzschild, taking isotropic coordinates so that the spatial
sections are conformally flat and the lines of force are normal to the
isopotentials. Isotropic coordinates for Schwarzschild correspond to
Poincar\'e coordinates on the hyperbolic space from which the optical
metric is constructed. The black region in the plots corresponds to
the interior of the black hole event horizon and we have returned the
asymptotically flat end to infinity.

\newpage
\begin{figure}[!h]
\begin{minipage}[t]{0.46\linewidth} % A minipage that covers half the page
\centering {\includegraphics[height=2.5in,
width=2.5in]{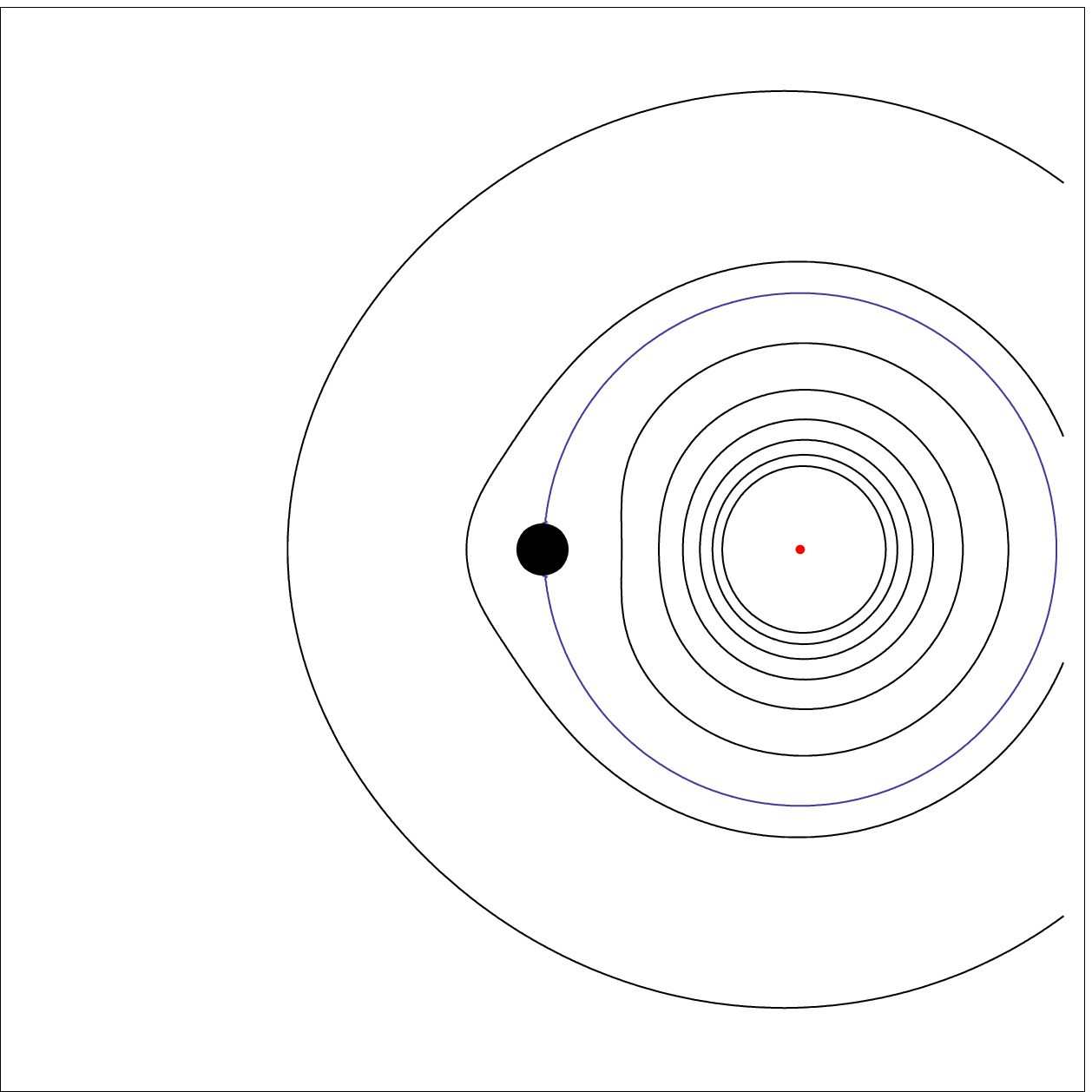}}
\end{minipage}
\hfill % To get a little bit of space between the figures
\begin{minipage}[t]{0.46\linewidth}
\centering  {\includegraphics[height=2.5in,
width=2.5in]{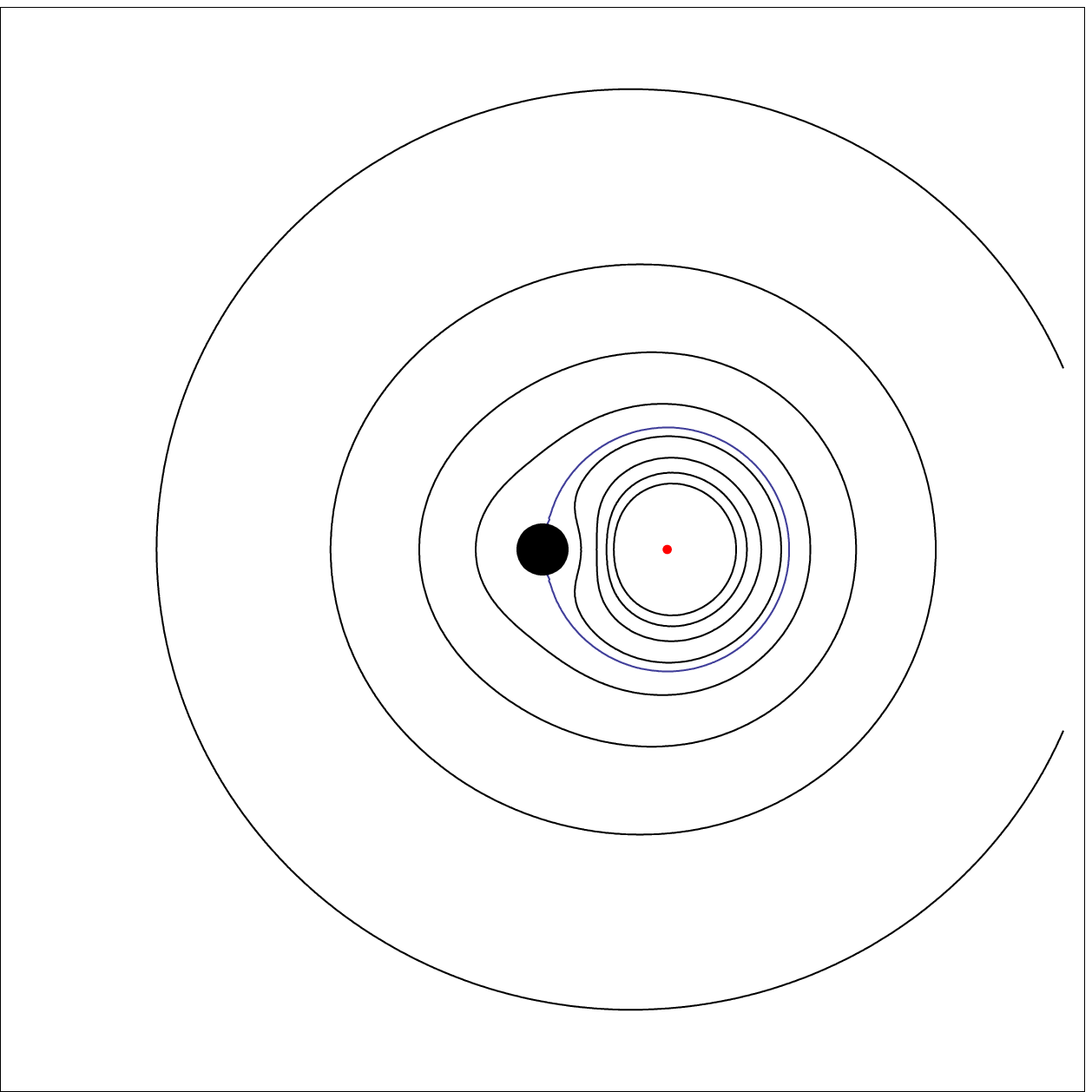}}
\end{minipage}
\end{figure}

\begin{figure}[!h]
\begin{minipage}[t]{0.46\linewidth} % A minipage that covers half the page
\centering {\includegraphics[height=2.5in,
width=2.5in]{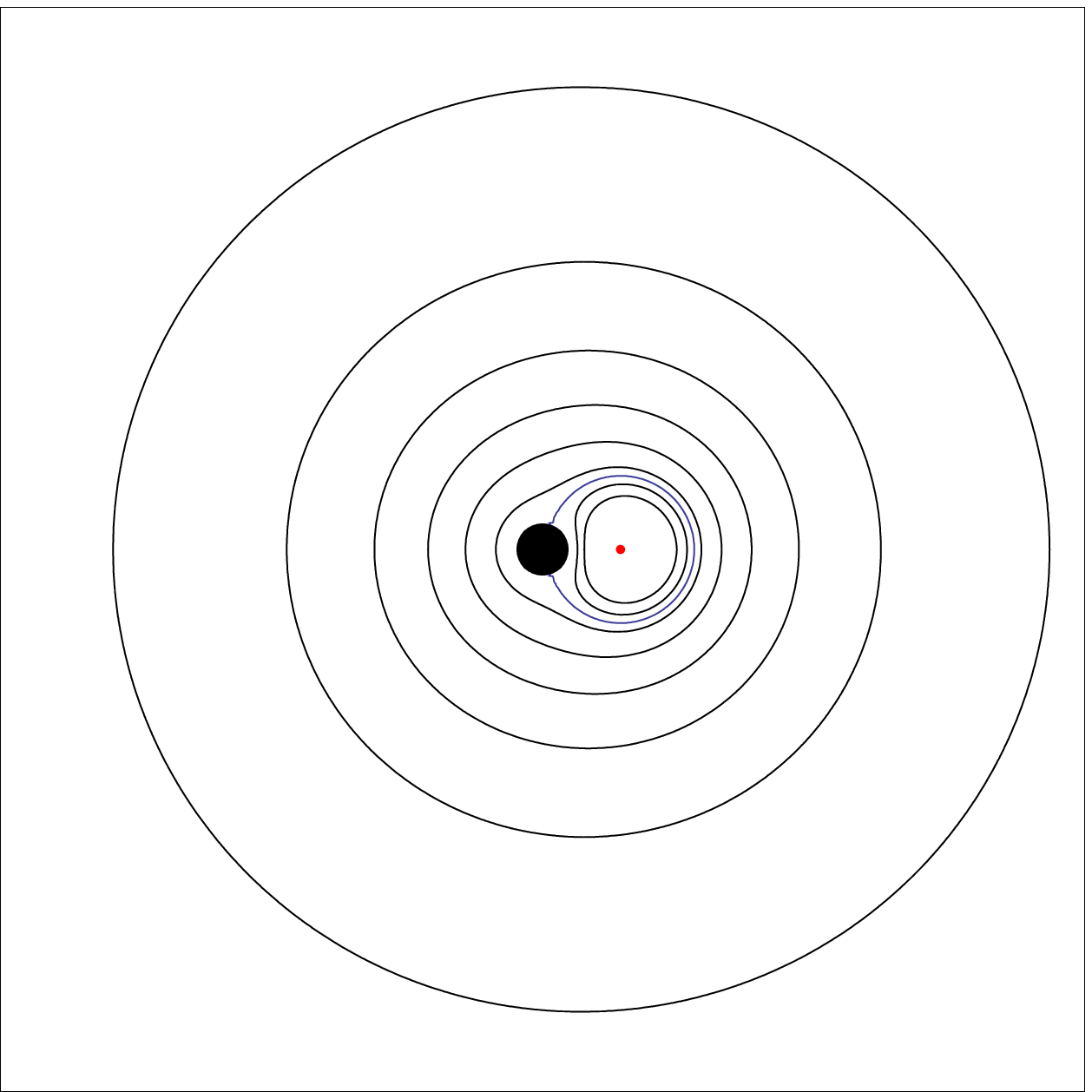}}
\end{minipage}
\hfill % To get a little bit of space between the figures
\begin{minipage}[t]{0.46\linewidth}
\centering  {\includegraphics[height=2.5in,
width=2.5in]{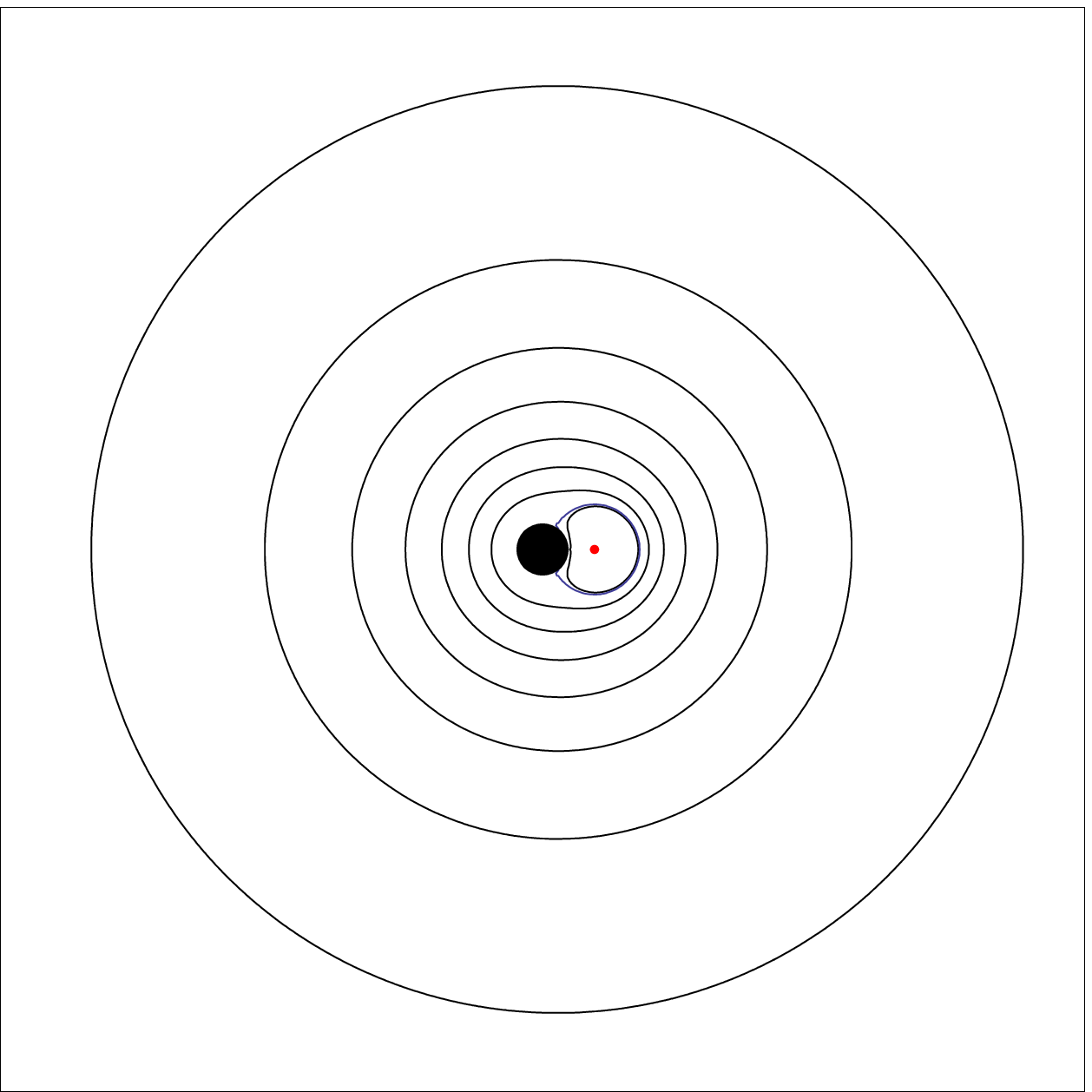}}
\end{minipage}
\end{figure}

\begin{figure}[!h]
\centering {\includegraphics[height=2.5in,
width=2.5in]{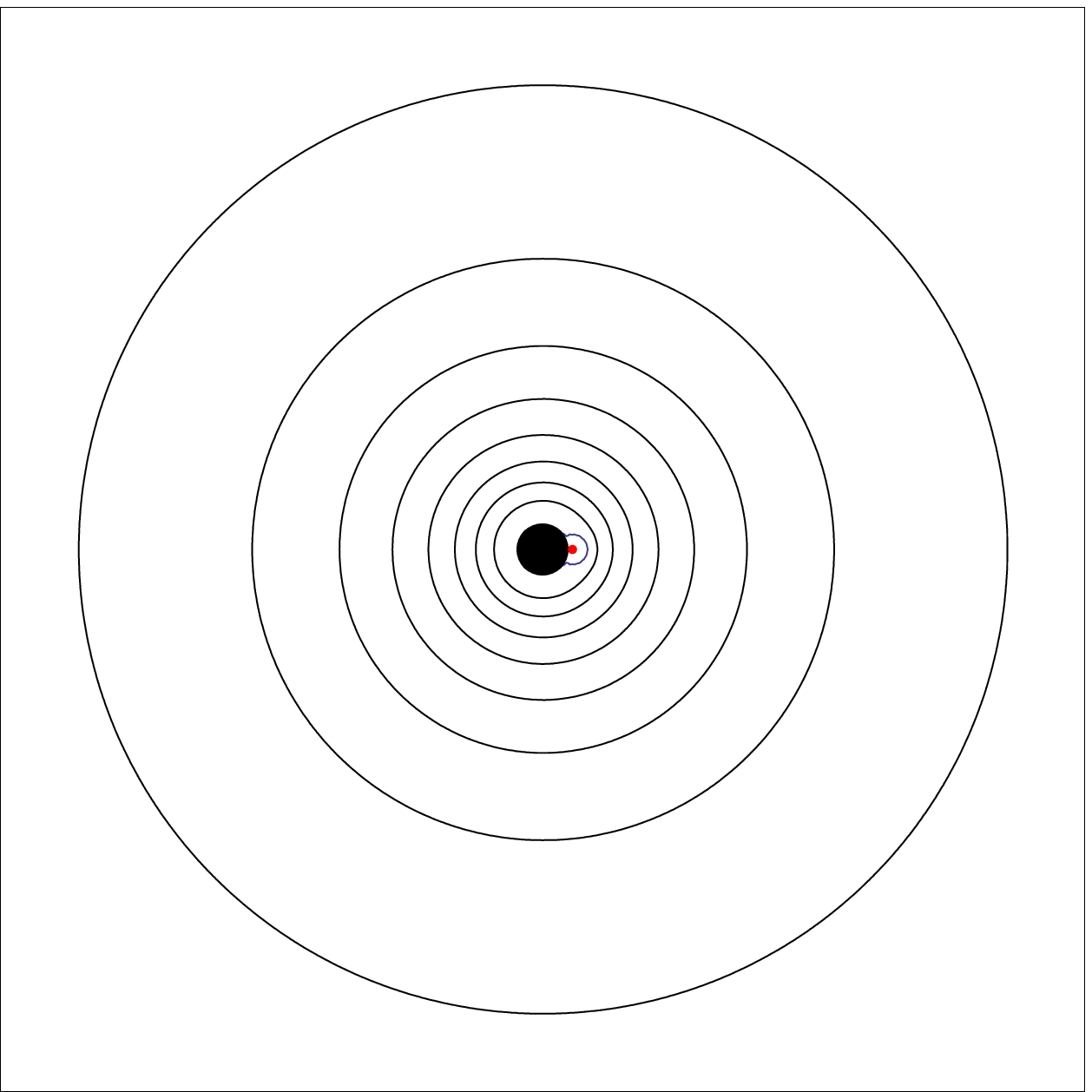}}
\end{figure}
{\centering Figure 3: Plots showing the equipotentials as a point charge is lowered
  into a Schwarzschild black hole in isotropic coordinates. The
  horizon is located at the boundary of the black disc and the point
  charge is red. The blue contour is the equipotential of the horizon.}

\subsubsection{Scalar Charge} \label{scalfield}

We will now show how to treat exactly a static massless scalar field
in the Schwarzschild or Reissner-Nordstr\"om backgrounds. We will once
again make use of the optical metric and the relationship between this
metric and the hyperbolic metric. The main result we shall require is
summarised as
\begin{lemma}
If $g$ is a scalar flat static metric of the form:
\begin{equation}
g = \Omega^2 \gopt = \Omega^2 (-dt^2 + H^{4} h),
\end{equation}
with $h$ the metric on $\Ht$ with radius $1$ and $H$ harmonic on
$\Ht$ then $\Omega$ must satisfy
\begin{equation}
\Delta_h(\Omega H) + \Omega H = 0.
\end{equation}
Further, if $\psi =
\Omega^{-1} H^{-1} \Phi$, then
\begin{equation}
\Box_g \psi = \Omega^{-3} \left[-\frac{1}{H} \pdd{\Phi}{t} +
  \frac{1}{H^5} \left(\Delta_h \Phi + \Phi \right) \right].
\end{equation}
\end{lemma}
\begin{proof}
This follows from standard identities for conformal transformations.
\end{proof}

We would like to calculate the field $G(x, x_0)$ at $x$ due to a unit static
scalar charge at $x_0$. For a general moving point charge, $G$
satisfies:
\begin{equation}
\Box_g G(x, x'(\tau)) = \int_{-\infty}^\infty \delta_{(g)}(x,
x'(\tau)) d\tau\, ,
\end{equation}
where $\tau$ is the proper time along the worldline $x'(\tau)$ of the
charged particle. Assuming that this particle is static, we find using
the Lemma above that $G(x, x_0) = \Omega^{-1}(x) H^{-1}(x)
\Phi(x,x_0)$ where $\Phi$ satisfies
\begin{equation}
\Delta_h \Phi(x, x_0) + \Phi(x, x_0)  =  H^{-1}(x_0) \delta_{(h)}(x, x_0).
\end{equation}
It is convenient once again to make use of the $SO(3,1)$ invariance of
hyperbolic space in order to solve this equation. We first seek
solutions to the simpler equation
\begin{equation}
\Delta_h \Phi_O(x) + \Phi_O(x)  = \delta_{(h)}(x, O)
\end{equation}
subject to the condition that $\Phi$ and $d \Phi$ are bounded in the
metric induced by $h$ as
$D(x,O) \to \infty$. One finds that the solution is related to the
metric functions for Reissner-Nordstr\"om according to:
\begin{equation}
\Phi_O(x) = \frac{1}{4\pi \sqrt{\mu}} \Omega(x) H(x).
\end{equation}
Unlike in the scalar charged case, there is no arbitrary
constant. Using the identities above, we find that the field due to a
static unit
scalar charge at a point $x_0$ is given by:
\begin{equation}
G(x,x_0) = \frac{\Phi_O(T_{x_0}(x)) }{H(x) H(x_0) \Omega(x)}\, ,
 \label{fulGscal}
\end{equation}
where $T_0$ is defined as in the previous section.

We may once again consider the behaviour of $G(x, x_0)$ as the scalar
charge approaches the horizon. We find that $G$ and its derivatives
fall off like $e^{-D( x_0,O)}$ as $D(x_0, O) \to \infty$. Unlike the
case of an electric charge, there is no residual monopole term, so the black
hole does not become charged. Thus, we see precisely how massless
scalar hair is shed as a point scalar charge is lowered into a black
hole and it it as predicted by our approximate argument given above.

\section{Conclusion}

We have seen how it is possible to make use of the universal asymptotics of the optical metric near a Killing horizon to study physical problems in this region. We have presented a method of studying null geodesics based on the Gauss-Bonnet theorem which directly links the negative curvature of the optical geometry to physical lensing scenarios. We have re-derived classic results about the loss of `hair' as objects fall into a black hole in a simplified manner and by making use of the universality of the near horizon optical metric, extended these results to apply beyond the Schwarzschild case where they were first investigated.

\appendix
\section{Integration on $\mathbb{CP}^1$}

In section \ref{phys} we found that the space of solutions to Dirac's equation on $\mathbb{R}_t\times \Ht$ could be identified with $\mathbb{R}_+ \times \mathbb{CP}^1$ where the $\mathbb{CP}^1$ arose by identifying Weyl spinors which were complex multiples of one another. In subsequent calculations it was necessary to integrate over this space of solutions in a Lorentz invariant fashion. The aim of this appendix is to explain how this is possible.

There are two key observations to be made. Firstly it should be noted that the space of Weyl spinors caries a natural $2$-form defined by:
\begin{equation}
\mu[\chi] = 2 i \epsilon_{\alpha \beta} \chi^\alpha d\chi^\beta \wedge \epsilon_{\dot{\alpha} \dot{\beta}} \bar{\chi}^{\dot{\alpha}} d\bar{\chi}^{\dot{\beta}}.
\end{equation}
This is Lorentz invariant by construction. 

Secondly we may represent $\mathbb{CP}^1$ as a smooth $2$-dimensional surface in $\mathbb{C}^2$, $\Sigma$ where we assume that for almost every $[\chi] \in \mathbb{CP}^1$ there is exactly one point $\tilde{\chi}\in \Sigma$ such that $[\chi] = [\tilde{\chi}]$. Since we are interested in integrating over $\mathbb{CP}^1$ it doesn't matter if this fails to be true for some set of measure zero. Suppose now that we chose a different surface $\Sigma'$. In order that this fulfils the requirements to represent $\mathbb{CP}^1$ there must exist some smooth function $\lambda: \mathbb{C}^2 \to \mathbb{C}$ such that for almost every point $\chi \in \Sigma$, $\lambda(\chi) \chi \in \Sigma'$. In other words we may, by extending the domain of $\phi$ if necessary define a local diffeomorphism
\begin{eqnarray}
\phi : U \subset \mathbb{C}^2 &\to& U' \subset\mathbb{C}^2 \nonumber \\
\chi &\mapsto& \lambda(\chi) \chi
\end{eqnarray}
such that $\phi(\Sigma) = \Sigma'$ up to a set of measure zero. One may verify that
\begin{equation}
\phi^*\mu = \abs{\lambda}^4 \mu.
\end{equation}
Thus if we have a function $f : \mathbb{C}^2 \to \mathbb{C}$ which is a scalar under Lorentz transformations and which satisfies $f(\lambda \chi) = \abs{\lambda}^{-4} f(\chi)$ then the integral
\begin{equation}
\int_\Sigma f \mu
\end{equation}
is independent of which surface in $\mathbb{C}^2$ we use to represent $\mathbb{CP}^1$. Suppose now that $f = f(\chi, X^i)$ where $X^i$ are some vectors in $\mathbb{E}^{3,1}$ and such that
\begin{equation}
f(\rho^s_\Lambda \chi, \rho^v_\Lambda X^i) = f(\chi, X^i),
\end{equation}
where $\rho^s, \rho^v$ are the spinor and vector representations of the Lorentz transformation $\Lambda$ respectively. If we pick a surface $\Sigma$ which represents $\mathbb{CP}^1$, we may define a function
\begin{equation}
I(X^i) = \int_\Sigma f(\chi, X^i) \mu = \int_\Sigma f(\rho^s_\Lambda \chi, \rho^v_\Lambda X^i) \mu.
\end{equation}
If $\phi_\Lambda$ is the function on $\mathbb{C}^2$ defined by left multiplication by $\rho^s_\Lambda$ then we may use the Lorentz invariance of the measure $\mu$ to write
\begin{eqnarray}
I(X^i) &=& \int_\Sigma \phi_\Lambda^*(f(\chi, \rho^v_\Lambda X^i) \mu) = \int_{\phi_\Lambda(\Sigma)} f(\chi, \rho^v_\Lambda  X^i) \mu  \nonumber\\ &=& \int_\Sigma f(\chi, \rho^v_\Lambda X^i) \mu = I(\rho^v_\Lambda X^i),
\end{eqnarray}
where we have made use of the Lorentz invariance of $f$ and $\mu$, together with the independence of the integral on the choice of representative of $\mathbb{CP}^1$. Thus the integral is a Lorentz scalar a fact which we make use of in section \ref{neutf} to calculate the integral (\ref{neutint})

As an example, we may take $\Sigma = \{ (1, z)^t/(1+\abs{z}^2)^{1/2} : z \in \mathbb{C} \}$ which covers all of $\mathbb{CP}^1$ except one point. We find then that
\begin{equation}
\mu |_{{}_\Sigma} = \frac{2 i}{\left(1+\abs{z}^2\right)^2} dz \wedge d\bar{z},
\end{equation}
the standard measure on the sphere under stereographic projection. We use this fact in calculating the neutrino mediated force between electrons.

\end{document}